\documentclass[11pt, a4 paper]{article}
\usepackage{fancyhdr}                   
\usepackage[Sonny]{fncychap}            
\usepackage[ruled,vlined,linesnumbered]{algorithm2e}
\usepackage{graphicx}
\usepackage{amsmath}
\usepackage{amssymb}
\usepackage{latexsym}
\usepackage{eqlist}                     
\usepackage{bookmark}
\usepackage{authblk}
\usepackage{geometry}
\usepackage{apacite}
\usepackage{mathtools}
\usepackage{listings}
\usepackage{graphics}
\usepackage{lscape}
\usepackage{longtable}
\usepackage[table,xcdraw]{xcolor} 
\usepackage{caption}
\usepackage{setspace}
\usepackage{booktabs}
\usepackage[T1]{fontenc}
\usepackage[utf8]{inputenc} 
\usepackage{lmodern}   
\usepackage{hyperref}
\usepackage{amsthm}
\usepackage{enumitem}
\usepackage{xr}
\usepackage{titlesec}
\usepackage[nameinlink,capitalize]{cleveref}
\newtheorem{theorem}{Theorem}
\newtheorem{lemma}{Lemma}

\title{Testing for Smooth Structural Change in Cointegrated Systems}
\author[1]{Haofeng Liao}
\author[1]{Xing Wang\thanks{Corresponding author. Email: \href{mailto:xing.wang@durham.ac.uk}{xing.wang@durham.ac.uk}}}
\date{August 2026}

\affil[1]{\textit{Department of Economics, Durham University}}
 
\begin{document}
\begin{spacing}{1.5}

\maketitle

\begin{abstract}

This paper develops an econometric framework for analysing smooth structural change in cointegrated systems following a known intervention time. We consider a vector error-correction model in which the cointegration rank and the pre-intervention cointegrating structure are identified from a stable pre-intervention subsample. After the intervention, both the adjustment coefficients and the cointegrating vectors are allowed to evolve smoothly as functions of rescaled time, which are estimated using kernel-weighted local reduced-rank methods. The analysis is formulated directly in a cointegrated VAR/VECM system, which preserves the treatment of long-run relations and short-run error-correction dynamics. By working with the decomposition $\Pi(\delta)=\alpha(\delta)\beta(\delta)'$, the method separates changes in the equilibrium relation from those in the speed of adjustment. We also provide two tests for the parameter consistency and the post-intervention parameter smoothness respectively. An empirical application to energy market, foreign-exchange, and gold-market index around the 24 February 2022 Russia's invasion of Ukraine illustrates how the proposed approach distinguishes between a discrete regime shift and smooth post-intervention evolution. The results suggest that cointegrating relation among the price of Brent crude oil, the spot exchange rate (USD/EUR), and the Credit Suisse NASDAQ Gold Price Index has smoothly changed after the outbreak of war, instead of a constant long-run conintegration system in the pre-intervention period.\\


\end{abstract}

\newpage

\section{Introduction}

Cointegration analysis is a standard tool for modelling long-run relations among nonstationary variables. The basic idea, due to Granger, is that some linear combinations of integrated $I(1)$ series may be stationary. The work of \shortciteA{engle1987co} and the likelihood-based approach of
\shortciteA{johansen1991estimation} provide the basis for much of the empirical work on cointegrated VARs and vector error-correction models. In these classical models, the cointegrating relation is usually taken to be constant over the sample. This assumption is convenient, but it can be restrictive in applications where the economic environment changes. Monetary-policy shifts, institutional reforms, financial crises, or geopolitical shocks may alter either the long-run relation itself or the speed at which variables adjust to it. Ignoring such changes may lead to a misspecified long-run structure and unreliable inference.\\

A large literature has therefore considered structural change in cointegrated systems. \shortciteA{gregory1996residual} develop residual-based tests that allow for a one-time shift in the intercept and/or slope of the cointegrating relation. Their results show that cointegration may be difficult to detect when a break in the long-run relation is omitted. In a VECM setting, \shortciteA{seo1998tests} derive Lagrange multiplier tests for structural change in both the cointegrating vector and the adjustment coefficient, with the break date treated as unknown. \shortciteA{kejriwal2010testing} extend this line of work by proposing a sup-Wald procedure for an unknown number of breaks in a cointegrating system. These approaches are useful when structural change is described by a finite number of discrete shifts. They are less satisfactory when the economic transition is gradual. In that case, a multiple-break specification may approximate smooth adjustment by a sequence of small jumps, which can be both statistically inefficient and difficult to interpret economically.\\

An alternative approach is to let the cointegrating parameters vary smoothly over time. \shortciteA{park1999cointegrating} study cointegrating regressions with time-varying coefficients using a sieve approximation and canonical cointegrating regressions. \shortciteA{juhl2005functional} consider cointegrating coefficients that follow a specified functional form. These papers show that gradual parameter change can be incorporated into cointegration models without imposing a small number of break dates. A limitation of such early approaches is that the form of time variation is often specified in advance. If the true path differs from the assumed trend or polynomial form, the resulting model may still be misspecified. \shortciteA{bierens2010time} propose a time-varying VECM in which the cointegrating vectors are smooth functions of time, approximated by Chebyshev time polynomials. Their model contains the standard Johansen specification as a special case and yields a likelihood-ratio test of time-invariant cointegration against a smooth time-varying alternative. This provides an important system-based treatment of time-varying cointegration.\\

Kernel-based methods offer a more local way to estimate smoothly changing cointegrating relations. \shortciteA{phillips2017estimating} develop kernel-based inference for cointegration models with time-varying structural
coefficients. This type of approach is attractive because it avoids imposing a fixed break structure or a low-dimensional global time trend. However, in many empirical settings the timing of a major event is known in advance. The research question is whether the long-run structure after the event departs from its pre-event benchmark, and whether this departure is gradual or abrupt.\\

This paper develops an event-anchored approach to time-varying cointegration. The model is designed for settings in which a major external intervention occurs at a known date $t^\ast$. Before $t^\ast$, the system is assumed to be a stable cointegrating system. The pre-intervention sample is therefore used to estimate the baseline rank, cointegrating vector, and
adjustment coefficient. After $t^\ast$, the long-run matrix is allowed to change with rescaled time $\delta=t/T$. In particular, both the cointegrating vector $\beta(\delta)$ and the adjustment coefficient $\alpha(\delta)$ may evolve over the post intervention period. This specification is intended to capture gradual structural adjustment after an external shock, for example when market participants learn about and adapt to a new policy environment or a new geopolitical regime.\\

The interpretation of the post-intervention relation is different from that in a time-invariant cointegration model. In the standard VECM, the matrix $\Pi=\alpha\beta'$ represents the long-run equilibrium relation and the adjustment mechanism. After a major intervention, the same variables may still share common stochastic trends, but the equilibrium relation or the adjustment
speed may no longer be the same as before. The pre-intervention estimate should therefore be viewed as a benchmark, rather than as a fixed relation imposed on the whole sample. A post intervention change in $\Pi(\delta)$ is interpreted
as evidence that the pre-intervention long-run structure has become unstable. This interpretation is consistent with tests for parameter instability in cointegrating regressions \shortcite{hansen1992parameter} and with cointegration tests allowing for regime shifts in the intercept or slope of the
cointegrating relation \shortcite{gregory1996residual}.\\

Estimation is carried out by local Johansen-type procedures on the post-intervention sample. At each grid point, observations close to that point receive larger kernel weights, producing local estimates of the long run matrix and its decomposition into $\alpha(\delta)$ and $\beta(\delta)$. The kernel weights define a local estimating window for the time-varying VECM parameters. The approach therefore relies on the usual local-smoothing conditions: the bandwidth shrinks on the rescaled time domain, the effective local sample size increases, and the local cointegration rank remains well defined.\\

The contribution of the paper is as follows. First, compared with discrete-break methods such as \shortciteA{seo1998tests}, the proposed model does not require all post-intervention changes to be captured by a single break. Such break-based methods are useful when the change is sudden and permanent, but they may be less suitable when the effect of an intervention develops gradually over time. In the present framework, the intervention date $t^*$ is treated as known, and the focus is on how the long-run structure evolves after that date. The model allows the post-intervention parameters to change smoothly, while the second test also checks whether the smooth path is interrupted by a jump or kink. Second, compared with single-equation time-varying cointegrating regressions such as \shortciteA{park1999cointegrating}, the analysis is carried out in a cointegrated VAR/VECM system. Instead, the model keeps the joint structure of the system, including both the long-run relation and the short-run error-correction dynamics. An intervention may change the long-run equilibrium relation, but it may also change how quickly each variable responds to deviations from that relation. The VECM setting therefore provides a more complete way to study post-intervention changes in a multivariate system. Third, the paper uses the decomposition
$\Pi(\delta)=\alpha(\delta)\beta(\delta)'$
to interpret the source of structural change. Changes in $\beta(\delta)$ indicate changes in the long-run equilibrium relation, while changes in $\alpha(\delta)$ indicate changes in the adjustment mechanism. This decomposition makes it possible to examine whether the intervention mainly changes the equilibrium relation itself, the speed of adjustment toward that relation, or both. In this way, the method provides a more informative interpretation than tests that only report whether a structural break is present.\\

The remainder of the paper is organised as follows. Section \ref{Smooth changing model} introduces the model, assumptions, and estimation procedure. Section \ref{LST3} presents the asymptotic theory of the estimators. Section \ref{Testing algorithm3} develops the block-bootstrap-based testing procedure. Section \ref{sec:simulation3} reports the simulation results.
Section \ref{Empirical3} presents the empirical application. Section \ref{Conclusion3} concludes.

\section{Smooth changing cointegrating model} \label{Smooth changing model}

Classical cointegration models initially assume constant long-run relationships, but later work raised concerns about structural stability. \shortciteA{hansen1992testing} develops a Lagrange multiplier (LM) test based on fully-modified or dynamic OLS estimators to detect a one-time parameter shift in the cointegrating regression. This single-equation approach focuses on changes in the cointegrating vector while treating the residual process nonparametrically. Hansen’s test examines changes in the cointegrating vector, not in the adjustment dynamics in the economic system. Therefore, the instability in adjustment coefficient (error-correction speeds) would be ignored, and the test provides no information on the timing or nature of the change beyond rejecting stability. \\

\shortciteA{gregory1996residual} extend residual-based cointegration tests to allow for one-time changes at unknown time points in the cointegration system. They propose modified ADF and Phillips–Perron-type tests that consider the possibility of a single break in either the intercept and/or slope of the cointegrated system. Gregory and Hansen’s framework allows three types of structural change: a level shift in the cointegrating intercept, a level shift with trend, or a full regime shift where both intercept and slope vector can change once. This approach can detect cointegration that standard tests might miss when a regime shift occurs – for example, an economic policy change causing a discrete change in the long-run relation. Although the break date is treated as unknown and searched over a trimmed set of candidate dates, the primary purpose of the test is to detect cointegration allowing for a regime shift, rather than to provide a detailed model of the post-break evolution. This method assumes any change is immediate and one-time, which is a simplification of reality. Only a single change is permitted, so multiple structural changes or gradual parameter drifts cannot be captured.\\

\shortciteA{seo1998tests} develop a system-wide test within the vector error-correction model (VECM) that could detect a one-time break in the cointegrating vector $\beta$ or in the adjustment coefficient $\alpha$. In \shortciteA{seo1998tests}'s setup, the one time structural change is defined in $\beta$ and $\alpha$ respectively:
\begin{align*}
&    \Delta X_t  = \alpha \begin{pmatrix} I_r \\ \beta+\phi \{t\geq[n\tau ]+1\} \end{pmatrix} X_{t-k} + \sum^{k-1}_{i=1} \Gamma_i \Delta X_{t-i} +\mu + \varepsilon_t, \\&
\Delta X_t  = ( \alpha+ \psi \{t\geq[n\tau ]+1\} ) \begin{pmatrix} I_r \\ \beta \end{pmatrix} X_{t-k} + \sum^{k-1}_{i=1} \Gamma_i \Delta X_{t-i} +\mu + \varepsilon_t,
\end{align*}
where  $X_t$ is a $p$-dimensional $I(1)$ process, $\phi$ and $\psi$ are coefficient matrix and $\{\cdot \}$ is the indicator function, in which the break point $\tau$ is fixed until the optimal tests for unknown $\tau$ is defined. However, \shortciteA{seo1998tests} only consider the one-time change in cointegrating vector and adjustment coefficient after unknown time point, which is not able to capture a continuous structural change. Seo’s test in the Johansen VECM framework uses a Lagrange multiplier statistic and is implemented as a supremum (sup-LM) test over all possible unknown break point. Seo derives asymptotic distributions for sup-LM, mean-LM, and exp-LM test statistics in the cointegration setting. This approach allows the break point to be unknown and jointly tests the stability of the cointegrating vector $\beta$ and adjustment coefficient $\alpha$, thus providing a more comprehensive check for structural change in a VAR system. Therefore, the test can distinguish whether a structural change affects the long-run equilibrium relation, the short-run dynamics, or both. However, the asymptotic critical values are non-standard and must be simulated. Moreover, like Gregory–Hansen, Seo’s original formulation assumes at most one structural change. This cannot deal with more complex evolution or multiple regime changes, and the test’s power may degrade if change is gradual rather than immediate. \\

The most recent and comprehensive approach to smooth structural change in cointegration is the nonparametric kernel-based framework developed by \shortciteA{phillips2017estimating}. They consider a cointegration model in which the cointegrating matrix $\Pi(\cdot)$ is a function of time. They consider a cointegration model with time-varying coefficient functions:
\[
y_t=x'_t f(t/T)+u_t = x'_t f_t+u_t , \ t=
1,\cdots,n,
\]
where $f(\cdot)$ is a $d$-dimensional function of time,  $x_t$ is an $I(1)$ vector, and $u_t$ is a scalar process. The function $f(t/T)$ can be regarded as a weak trend function so that the model captures potential drifts in the cointegrating relationship between $y_t$ and $x_t$ over time. Instead of assuming a particular parametric form, they employ kernel Nadaraya–Watson local regression to estimate the coefficient functions directly from the data. This kernel-based method have good flexibility and generality: it can capture arbitrary smooth changes in the cointegrating relationship, and it accommodates endogenous regressors by extending the fully modified regression to a time-varying context. Their kernel estimator is bias-corrected (using a local linear adjustment) and achieves consistency even when regressors are $I(1)$ and endogenous. They also provide an asymptotic theory showing that the estimates of the time-varying coefficient converge at different rates in different directions of the parameter space. However, their framework is formulated as a time-varying cointegrating regression rather than as a VECM in which the long-run impact matrix is decomposed into an adjustment matrix and cointegrating vectors. Therefore, it does not directly distinguish changes in the equilibrium relation from changes in the speed of error correction. Furthermore, fully smooth time variation can be weakly linked to empirical events. In practice, people know the date of a major intervention (e.g.\ a policy regime change or a geopolitical shock) and would like to exploit this information explicitly. Existing smooth time-varying models typically treat parameter change as an inherent process operating throughout the full sample, without a well-defined pre- and post-intervention structure. This gap motivates an event-anchored framework that combines the interpretability of discrete-break models with the flexibility of functional-coefficient approaches. \\

In this paper, we propose a smooth structural-change framework for cointegrated systems with a known intervention time point. The cointegration rank is identified from a stable pre-intervention subsample, while both the cointegrating vectors and the adjustment coefficients in a VECM are allowed to evolve smoothly over the post-intervention period. Estimation is conducted via kernel-WLS and Johansen estimation. We develop two complementary post-intervention tests: a baseline test that evaluates whether post-intervention parameters equal the pre-intervention benchmark, and a smoothness test that compares a pooled local-linear fit with split left- and right-sided local-linear fits, thereby detecting possible jumps or kinks in the estimated post-intervention long-run matrix.

\subsection{Model setup and identification}

Let $\{X_t\}_{t=1}^T$ be a $p$-dimensional $I(1)$ process, and let $t^*$ denote a known intervention date with corresponding sample fraction $\delta^{\ast}=t^*/T$. Consider a post-intervention vector error-correction representation of the form
\begin{equation}
\Delta X_t
=
\Pi(\delta)X_{t-k}
+
\sum_{i=1}^{k-1}\Gamma_i \Delta X_{t-i}
+
\mu
+
\varepsilon_t,
\qquad
t=t^*+1,\ldots,T,
\qquad
\delta=t/T,
\label{eq:post_vecm}
\end{equation}
where $\Gamma_1,\ldots,\Gamma_{k-1}$ and $\mu$ are assumed constant over the post-intervention period, while the long-run impact matrix $\Pi(\delta)$ is allowed to evolve smoothly with time. The innovation vector $\varepsilon_t$ admits a linear process representation
\begin{equation}
\varepsilon_t = \Phi(L)e_t = \sum_{j=0}^{\infty}\Phi_j e_{t-j},
\end{equation}
with $\{e_t\}$ i.i.d., $\mathbb{E}(e_t)=0$, $\mathbb{E}(e_te_t')=V_e>0$, $\mathbb{E}\|e_t\|^{4+\eta}<\infty$ for some $\eta>0$, and
\begin{equation}
\sum_{j=0}^{\infty} j\|\Phi_j\| < \infty.
\end{equation}

The cointegration rank is assumed to be fixed and equal to $r$ on $[\delta^{\ast},1]$, which is identified in the pre-intervention period, so that for each $\delta\in[\delta^{\ast},1]$,
\begin{equation}
\Pi(\delta)=\alpha(\delta)\beta(\delta)',
\qquad
\alpha(\delta)\in\mathbb{R}^{p\times r},
\qquad
\beta(\delta)\in\mathbb{R}^{p\times r},
\qquad
\text{rank}\,\Pi(\delta)=r.
\label{eq:Pi_factorisation}
\end{equation}

In the pre-intervention period, the system is assumed to be stable with constant parameters:
\begin{equation}
\Delta X_t
=
\alpha_0\beta_0'X_{t-k}
+
\sum_{i=1}^{k-1}\Gamma_i \Delta X_{t-i}
+
\mu
+
\varepsilon_t,
\qquad
t=1,\ldots,t^*.
\label{eq:pre_vecm}
\end{equation}
Hence, $(\alpha_0,\beta_0)$ and the rank $r$ can be estimated from the pre-intervention subsample by Johansen's method \shortcite{johansen1988statistical}.\\

To identify the cointegrating vectors, we impose the standard Johansen normalisation
\begin{equation}
\beta(\delta)
=
\begin{pmatrix}
I_r\\
C(\delta)
\end{pmatrix},
\qquad
C(\delta)\in\mathbb{R}^{(p-r)\times r},
\label{eq:beta_normalisation}
\end{equation}
where $I_r$ is the $r\times r$ identity matrix. Under \eqref{eq:beta_normalisation}, smooth post-intervention structural change is fully characterised by the functional parameters $\alpha(\delta)$ and $C(\delta)$.\\

For the smooth-transition benchmark, we impose continuity at the intervention boundary:
\begin{equation}
\alpha(\delta^{\ast})=\alpha_0,
\qquad
\beta(\delta^{\ast})=\beta_0,
\qquad
C(\delta^{\ast})=C_0.
\label{eq:boundary_condition}
\end{equation}
This condition is used to formalise gradual adjustment after the event. It can be relaxed when the empirical question allows an immediate level shift at $t^*$, in which case the post-intervention functions are interpreted as right-continuous limits after the intervention.

\subsection{Kernel estimation of the functional cointegrating structural coefficients}

Fix an interior post-intervention point $\delta_0\in[\delta^{\ast},1]$ and define the kernel weights
\begin{equation}
K_{t,h}(\delta_0)
=K\!\left(\frac{t/T-\delta_0}{h}\right)\mathbf{1}\{t> t^*\}, 
\label{eq:kernel_weights}
\end{equation}
where $K(\cdot)$ is some kernel function and $h$ is a bandwidth. For sufficiently small $h$, the coefficient functions are approximately constant in a neighbourhood of $\delta_0$, so that \eqref{eq:post_vecm} can be locally approximated by
\begin{equation}
\Delta X_t
\approx
\Pi(\delta_0)X_{t-k}
+
\sum_{i=1}^{k-1}\Gamma_i \Delta X_{t-i}
+
\mu
+
\varepsilon_t,
\label{eq:local_vecm}
\end{equation}
where $\widehat{\Pi}$ is estimated by the Johansen-based Kernel-WLS estimation as follows.\\

To estimate $\widehat{\alpha}(\delta)$ and $\widehat{\beta}(\delta)$, we follow the Johansen's estimation \shortcite{johansen1988statistical}. Specifically, let
\begin{equation}
Z_t
=
\bigl(\Delta X_{t-1}',\ldots,\Delta X_{t-k+1}',1\bigr)'
\in\mathbb{R}^{(k-1)p+1}
\label{eq:short_run_regressors}
\end{equation}
collect the short-run regressors and deterministic term. To isolate the local long-run relation, partial out $Z_t$ from both $\Delta X_t$ and $X_{t-k}$ using the same kernel weights as those employed in the local reduced-rank problem. Define
\begin{align}
\widehat P_{ZZ}(\delta_0)
&=\frac{1}{\sum_{t=t^*+1}^T K_{t,h}(\delta_0)}
\sum_{t=t^*+1}^T
K_{t,h}(\delta_0)\,Z_tZ_t',
\label{eq:SZZ}
\\
\widehat P_{Z0}(\delta_0)
&=\frac{1}{\sum_{t=t^*+1}^T K_{t,h}(\delta_0)}
\sum_{t=t^*+1}^T
K_{t,h}(\delta_0)\,Z_t\Delta X_t',
\label{eq:SZ0}
\\
\widehat P_{Z1}(\delta_0)
&=\frac{1}{\sum_{t=t^*+1}^T K_{t,h}(\delta_0)}
\sum_{t=t^*+1}^T
K_{t,h}(\delta_0)\,Z_tX_{t-k}'.
\label{eq:SZ1}
\end{align}
The local weighted least-squares regression
coefficients are then
\begin{equation}
\widehat B_0(\delta_0)
=
\widehat P_{ZZ}(\delta_0)^{-1}\widehat P_{Z0}(\delta_0),
\qquad
\widehat B_1(\delta_0)
=
\widehat P_{ZZ}(\delta_0)^{-1}\widehat P_{Z1}(\delta_0),
\label{eq:local_regression_coefficients}
\end{equation}
with corresponding local residuals
\begin{equation}
R_{0t}(\delta_0)
=
\Delta X_t-\widehat B_0(\delta_0)'Z_t,
\qquad
R_{1t}(\delta_0)
=
X_{t-k}-\widehat B_1(\delta_0)'Z_t.
\label{eq:local_residuals}
\end{equation}

Using these locally residualised variables, define the local covariance matrices
\begin{equation}
\widehat S_{ij}(\delta_0)
=
\frac{1}{\sum_{t=t^*+1}^T
K_{t,h}(\delta_0)}
\sum_{t=t^*+1}^T
K_{t,h}(\delta_0)\,
R_{it}(\delta_0)R_{jt}(\delta_0)',
\qquad
i,j\in\{0,1\}.
\label{eq:local_covariances}
\end{equation}

For each fixed $\delta$, we estimate $\alpha(\delta)$ and $\beta(\delta)$ (write $\alpha, \beta$ for convenient representation) by minimising the reduced-rank WLS criterion
\begin{equation}
\mathcal L_{\delta}(\alpha,\beta)
=
\text{tr}
\Big[
\widehat S_{00}(\delta_0)
-
\alpha\beta'\widehat S_{10}(\delta_0)
-
\widehat S_{01}(\delta_0)\beta\alpha'
+
\alpha\beta'\widehat S_{11}(\delta_0)\beta\alpha'
\Big].
\label{eq:criterion}
\end{equation}
For fixed $\beta$, the minimiser in $\alpha$ is
\begin{equation}
\widehat\alpha(\delta_0;\beta)
=
\widehat S_{01}(\delta_0)\beta
\bigl[\beta'\widehat S_{11}(\delta_0)\beta\bigr]^{-1}.
\label{eq:alpha_given_beta}
\end{equation}
Substituting \eqref{eq:alpha_given_beta} into \eqref{eq:criterion} yields the profile criterion
\begin{equation}
\mathcal L_{\delta_0}(\beta)
=
\text{tr}\!\left(
\widehat S_{00}(\delta_0)
-
\widehat S_{01}(\delta_0)\beta
\bigl[\beta'\widehat S_{11}(\delta_0)\beta\bigr]^{-1}
\beta'\widehat S_{10}(\delta_0)
\right).
\label{eq:profile_criterion}
\end{equation}
Then $\widehat\beta(\delta_0)$ is obtained from the generalised eigenvalue problem (according to \shortciteA{johansen1988statistical})
\begin{equation}
\det\!\Big(
\lambda\,\widehat S_{11}(\delta_0)
-
\widehat S_{10}(\delta_0)\widehat S_{00}(\delta_0)^{-1}\widehat S_{01}(\delta_0)
\Big)=0.
\label{eq:gevp}
\end{equation}
Let
\begin{equation}
\widehat\lambda_1(\delta_0)\geq\cdots\geq\widehat\lambda_p(\delta_0)\geq 0
\end{equation}
denote the ordered generalised eigenvalues, and let
\begin{equation}
\widehat \xi_1(\delta_0),\ldots,\widehat \xi_r(\delta_0)
\end{equation}
be the corresponding eigenvectors associated with the $r$ largest eigenvalues. Then the local estimator of the cointegrating space is
\begin{equation}
\widetilde \beta(\delta_0)
=
\bigl(
\widehat v_1(\delta_0),\ldots,\widehat v_r(\delta_0)
\bigr),
\label{eq:beta_tilde}
\end{equation}
after normalisation, 
\begin{equation}
\widetilde \beta(\delta_0) \to \widehat \beta(\delta_0)= 
\begin{pmatrix}
I_r\\
\widehat C(\delta_0)
\end{pmatrix},
\label{eq:beta_hat}
\end{equation}
Then the local estimator of the adjustment matrix is
\begin{equation}
\widehat\alpha(\delta_0)
=
\widehat S_{01}(\delta_0)\widehat\beta(\delta_0)
\bigl[
\widehat\beta(\delta_0)'\widehat S_{11}(\delta_0)\widehat\beta(\delta_0)
\bigr]^{-1}.
\label{eq:alpha_hat}
\end{equation}

In implementation, the local eigenvectors are only defined up to sign and, more generally, up to orthogonal rotation within the estimated $r$-dimensional cointegrating space. To obtain a smooth estimated path $\delta_0 \mapsto\widehat\beta(\delta_0)$, the local estimates are aligned sequentially across neighbouring evaluation points using the normalisation \eqref{eq:beta_normalisation} together with a sign/rotation-matching rule. This alignment is identificational and does not alter the estimated local cointegrating space itself.

\subsection {Remark on kernel estimation}

The local estimation step is based on the assumption that the long-run coefficient matrix evolves gradually over rescaled time. The kernel weights define a local estimating window around each grid point, so that observations closer to that point receive larger weights in the estimation of the local VECM. This is consistent with time-varying cointegration models in which the cointegrating relation changes smoothly over time, with the standard Johansen model obtained as the constant coefficient special case \shortcite{bierens2010time}. It is also similar to  cointegrating regressions with smoothly time-varying coefficients \shortcite{park1999cointegrating} and with kernel-based inference for time-varying coefficient cointegrating regressions involving multiple non-stationary regressors \shortcite{li2020kernel}. The validity of the local approach therefore relies on the usual joint conditions for non-stationary local estimation: the coefficient path must be sufficiently smooth within each
local window, the bandwidth must shrink on the rescaled time domain while the effective local sample size diverges, and the local cointegration rank must be well defined.\\

The use of kernel weights in the construction of
$\widehat P_{ZZ}(\delta_0)$, $\widehat P_{Z0}(\delta_0)$, $\widehat P_{Z1}(\delta_0)$, and hence in the local covariance matrices $\widehat S_{ij}(\delta_0)$, because the objects of interest are smooth time-indexed functions:
\[
\alpha(\delta_0),\qquad \beta(\delta_0),\qquad \Pi(\delta_0)=\alpha(\delta_0)\beta(\delta_0)'.
\]
The kernel $K_{t,h}(\delta_0)$ provides a smooth weighting tool over time. Observations with $t/T$ close to the evaluation point $\delta_0$ receive relatively large weights, while observations farther away receive smaller weights. Therefore, the estimator at $\delta_0$ is a local average of information around $\delta_0$, rather than an estimate based on a single observation or on the whole post-intervention sample.\\

As the evaluation point $\delta_0$ moves gradually over the post-intervention interval, the kernel window also moves gradually, and the estimators are computed from overlapping local samples. Hence the kernel imposes a smooth-in-time estimation structure. \\

This is consistent with the modelling assumption that the post-intervention structural change is gradual rather than discontinuous. Under the smoothness condition on $\alpha(\cdot)$ and $C(\cdot)$, for $t/T$ sufficiently close to $\delta$,
\[
\alpha(t/T)\approx \alpha(\delta),
\qquad
\beta(t/T)\approx \beta(\delta),
\qquad
\Pi(t/T)\approx \Pi(\delta).
\]
The kernel-weighted criterion therefore estimates the locally constant approximation to the smooth functional parameter in a VECM at the point $\delta$.\\

The same kernel weights are also used when partial out the short-run regressors $Z_t$. The matrices
\begin{align*}
\widehat P_{ZZ}(\delta_0)
&=\frac{1}{\sum_{t=t^*+1}^T K_{t,h}(\delta_0)}
\sum_{t=t^*+1}^T
K_{t,h}(\delta_0)\,Z_tZ_t',
\\
\widehat P_{Z0}(\delta_0)
&=\frac{1}{\sum_{t=t^*+1}^T K_{t,h}(\delta_0)}
\sum_{t=t^*+1}^T
K_{t,h}(\delta_0)\,Z_t\Delta X_t',
\\
\widehat P_{Z1}(\delta)
&=\frac{1}{\sum_{t=t^*+1}^T K_{t,h}(\delta_0)}
\sum_{t=t^*+1}^T
K_{t,h}(\delta_0)\,Z_tX_{t-k}'
\end{align*}
define the local weighted regressions of $\Delta X_t$ and $X_{t-k}$ on $Z_t$. Thus the residuals
\[
R_{0t}(\delta_0)=\Delta X_t-\widehat B_0(\delta_0)'Z_t,
\qquad
R_{1t}(\delta_0)=X_{t-k}-\widehat B_1(\delta_0)'Z_t
\]
represent the short-run-adjusted components of $\Delta X_t$ and $X_{t-k}$ within the same local neighbourhood of $\delta$. The short-run dynamics should be removed locally, using the same time weights as those used to estimate the long-run structure. If $Z_t$ were partialled out using unweighted or global regressions, the residuals would reflect an average short-run relationship over the entire post-intervention period, while the reduced-rank estimation of $\alpha(\delta)$ and $\beta(\delta)$ would remain local. Such a mismatch would contaminate the local covariance matrices and weaken the interpretation of the local eigenvalue problem.\\

After residualisation, the matrices
\[
\widehat S_{ij}(\delta)
=
\frac{1}{\sum_{t=t^*+1}^T
K_{t,h}(\delta_0)}
\sum_{t=t^*+1}^T
K_{t,h}(\delta_0)
R_{it}(\delta)R_{jt}(\delta)',
\qquad i,j\in\{0,1\},
\]
are the local kernel-smoothed analogues of Johansen's residual covariance matrices. They summarise the local relationship between the differenced process $R_{0t}(\delta_0)$ and the lagged level process $R_{1t}(\delta_0)$, after removing short-run components. The resulting local reduced-rank problem is therefore the natural kernel-weighted counterpart of the standard Johansen eigenvalue problem.\\

This construction follows the logic of the kernel-based time-varying cointegration literature. \shortciteA{phillips2017estimating} study cointegration models in which structural coefficients evolve smoothly over time and are estimated by nonparametric kernel methods. Their framework shows that kernel smoothing is a natural device for recovering smoothly time-varying cointegrating coefficients, while also highlighting that integrated regressors generate non-standard signal matrices and require careful asymptotic theory \shortcite{phillips2017estimating,li2020kernel}. The present paper adapts the same idea to a VECM setting in which the post-intervention long-run matrix is decomposed as $\Pi(\delta)=\alpha(\delta)\beta(\delta)'$. The kernel weights therefore serve three related purposes: (i) they localise estimation at each $\delta$; (ii) they smooth the estimated coefficient paths over time; (iii) they ensure that the partialling-out step and the reduced-rank estimation step are based on the same information set.

\subsection{Assumptions for identification and estimation}

\paragraph{Assumption 3.1(Smooth functional parameters).}
For $\delta\in[\delta^{\ast},1]$, the functional parameters $\alpha(\delta)$ and $C(\delta)$ are bounded and continuously differentiable, and there exists $\gamma\in(1/2,1]$ such that
\begin{equation}
\|\alpha(\delta+z)-\alpha(\delta)\| = O(|z|^\gamma),
\qquad
\|C(\delta+z)-C(\delta)\| = O(|z|^\gamma),
\qquad
z\to 0,
\end{equation}
uniformly in $\delta\in[\delta^{\ast},1]$.

\paragraph{Assumption 3.2 (Kernel and bandwidth).}
The kernel function $K(\cdot)$ is bounded, nonnegative, continuous, symmetric, integrates to one, and either has compact support or has exponentially decaying tails. The bandwidth satisfies
\begin{equation}
h\to 0,
\qquad
(T-t^*)h \to \infty
\end{equation}
as $T\to\infty$.

\paragraph{Assumption 3.3 (Post-intervention short-run stability).}
The short-run coefficients 
$$\Gamma_1,\ldots,\Gamma_{k-1}$$ 
and the intercept $\mu$ are constant over the post-intervention period. Hence, all smooth structural change is confined to the long-run component $\Pi(\delta)=\alpha(\delta)\beta(\delta)'$.

\paragraph{Assumption 3.4 (Normalisation and continuity at the boundary).}
The cointegrating vectors satisfy the normalisation \eqref{eq:beta_normalisation} for all $\delta\in[\delta^{\ast},1]$, and the boundary conditions \eqref{eq:boundary_condition} hold.\\

Under Assumptions 3.1--3.4, the local kernel estimands above is defined, the local cointegrating space can be identified, and the estimators $\widehat\alpha(\delta)$ and $\widehat\beta(\delta)$ provide a nonparametric characterisation of smooth post-intervention structural change in the cointegrated system.

\section{Large sample theory for the estimators}
\label{LST3}

This section develops asymptotic theory for the local kernel estimators. Fix an interior post-intervention point $\delta_0 \in (\delta^\ast,1)$, and let
\begin{equation}
N := T-t^*
\label{eq:post-intervention-sample-size}
\end{equation}
denote the effective post-intervention sample size.

Under the normalisation, write
\begin{equation}
\beta(\delta_0)
=
\begin{pmatrix}
I_r\\
C(\delta_0)
\end{pmatrix},
\qquad
C(\delta_0)\in\mathbb R^{(p-r)\times r},
\label{eq:normalised-beta-free-block}
\end{equation}
and define the parameter vector
\begin{equation}
\vartheta(\delta_0):=\operatorname{vec}\!\big(C(\delta_0)\big).
\label{eq:free-parameter-vector}
\end{equation}

For the asymptotic analysis, define the population local  estimation coefficients
\begin{equation}
B_0^\circ(\delta_0)
=
P^\circ_{ZZ}(\delta_0)^{-1}P^\circ_{Z0}(\delta_0),
\qquad
B_1^\circ(\delta_0)
=
P^\circ_{ZZ}(\delta_0)^{-1}P^\circ_{Z1}(\delta_0),
\label{eq:population-regression-coefficients}
\end{equation}
where
\begin{equation}
P^\circ_{ZZ}(\delta_0)
=
\mathbb{E}\!\left[
K_{t,h}(\delta_0)Z_tZ_t'
\right],
\label{eq:population-SZZ}
\end{equation}
\begin{equation}
P^\circ_{Z0}(\delta_0)
=
\mathbb{E}\!\left[
K_{t,h}(\delta_0)Z_t\Delta X_t'
\right],
\label{eq:population-SZ0}
\end{equation}
and
\begin{equation}
P^\circ_{Z1}(\delta_0)
=
\mathbb{E}\!\left[
K_{t,h}(\delta_0)Z_tX_{t-k}'
\right].
\label{eq:population-SZ1}
\end{equation}

The corresponding population local residuals are
\begin{equation}
R_{0t}^\circ(\delta_0)
=
\Delta X_t-B_0^\circ(\delta_0)'Z_t,
\label{eq:population-residual-y}
\end{equation}
and
\begin{equation}
R_{1t}^\circ(\delta_0)
=
X_{t-k}-B_1^\circ(\delta_0)'Z_t.
\label{eq:population-residual-x}
\end{equation}

\subsection{Local population equivalence}

\begin{lemma}[Local approximation and population residualisation]
\label{lem:local-population-equivalence}
Suppose Assumptions 3.1--3.4 hold. Then, for the fixed interior point $\delta_0$,

\begin{enumerate}
    \item[(i)] if $K_{t,h}(\delta_0)\neq 0$, then
    \begin{equation}
    \|\Pi(\delta)-\Pi(\delta_0)\|=O(h^\gamma);
    \label{eq:local-pi-approximation}
    \end{equation}

    \item[(ii)] the local regression parameters satisfy
    \begin{equation}
    \|\widehat B_0(\delta_0)-B_0^\circ(\delta_0)\|=o_p(1),
    \qquad
    \|\widehat B_1(\delta_0)-B_1^\circ(\delta_0)\|=o_p(1);
    \label{eq:local-regression-consistency}
    \end{equation}

    \item[(iii)] the population and local residuals are asymptotically equivalent:
    \begin{equation}
    \max_{t:\,K_{t,h}(\delta_0)\neq 0}
    \|R_{0t}(\delta_0)-R_{0t}^\circ(\delta_0)\|=o_p(1),
    \label{eq:local-y-residual-equivalence}
    \end{equation}
    \begin{equation}
    \max_{t:\,K_{t,h}(\delta_0)\neq 0}
    \|R_{1t}(\delta_0)-R_{1t}^\circ(\delta_0)\|=o_p(1).
    \label{eq:local-x-residual-equivalence}
    \end{equation}
\end{enumerate}
\end{lemma}
The proof is in the Appendix A.\\

Lemma \ref{lem:local-population-equivalence} is a bridge for the estimation and testing arguments. Part (i) shows that, within a bandwidth-$h$ neighbourhood of the interior point $\delta_0$, the time-varying long-run matrix $\Pi(\delta)$ is locally well approximated by its value at $\delta_0$, with approximation error of order $O(h^\gamma)$. This is the basic smoothness property that justifies the local kernel-weighted approximation. Parts (ii) and (iii) then show that replacing the population residualisation objects by their feasible sample analogues does not affect the asymptotic argument, because the estimators are consistent and the induced residual differences are uniformly negligible on the local window. As a result, the subsequent derivation will work with the population residuals as an asymptotically equivalent representation, which simplifies the derivation of the local estimator's limit behaviour. 

\subsection{Consistency of the coefficient estimators}

\begin{theorem}[Consistency]
\label{thm:consistency}
Suppose the conditions of Lemma \ref{lem:local-population-equivalence} hold. In addition, let
\[
\lambda_1^{\circ}(\delta_0)\ge \cdots \ge \lambda_p^{\circ}(\delta_0)
\]
denote the ordered generalised eigenvalues solution of the population matrix pair
\[
\left(
S_{10}^{\circ}(\delta_0)
S_{00}^{\circ}(\delta_0)^{-1}
S_{01}^{\circ}(\delta_0),
\,
S_{11}^{\circ}(\delta_0)
\right),
\]
where $S_{ij}^{\circ}(\delta_0)$ denotes the corresponding population covariance matrix. Suppose that there exists a constant $c_\lambda>0$ such that
\[
\lambda_r^{\circ}(\delta_0)-\lambda_{r+1}^{\circ}(\delta_0)\ge c_\lambda .
\]
This eigenvalue separation condition ensures that the $r$-dimensional local cointegrating space is uniquely identified by avoiding weak identification cases in which the $r$th and $(r+1)$th eigen-directions are asymptotically indistinguishable. It is therefore needed to apply eigen-space perturbation arguments and to obtain consistency of the normalised estimator $\widehat\beta(\delta_0)$. Then, after normalisation and sign/rotation alignment,
\begin{equation}
\widehat\beta(\delta_0)\xrightarrow{p}\beta(\delta_0),
\qquad
\widehat\alpha(\delta_0)\xrightarrow{p}\alpha(\delta_0).
\label{eq:consistency}
\end{equation}
Equivalently, for the parameter block $C(\delta_0)$ in $\beta(\delta_0)$,
\begin{equation}
\widehat C(\delta_0)\xrightarrow{p}C(\delta_0).
\label{eq:free-block-consistency}
\end{equation}
\end{theorem}
The proof is in the Appendix A.

\subsection{Asymptotic convergence}
\label{asymdist3}

This subsection develops asymptotic convergence for the local estimators at a fixed interior post-intervention point $\delta_0\in(\delta^\ast,1)$. Under the normalisation
\begin{equation}
\beta(\delta_0)
=
\begin{pmatrix}
I_r\\
C(\delta_0)
\end{pmatrix},
\qquad
C(\delta_0)\in\mathbb R^{(p-r)\times r},
\label{eq:beta-normalisation}
\end{equation}
define
\begin{equation}
\vartheta(\delta_0):=\operatorname{vec}\!\big(C(\delta_0)\big)\in\mathbb R^{m},
\qquad
m:=r(p-r),
\label{eq:vartheta}
\end{equation}
and let $\widehat\vartheta(\delta_0)$ denote the corresponding estimator induced by $\widehat\beta(\delta_0)$.

\paragraph{Assumption 3.5 (Post-intervention functional central limit theorem).}
Following \shortciteA{phillips2017estimating}, on the post-intervention subsample, let
\begin{equation}
\varepsilon_t=\Phi(L)e_t=\sum^\infty_{\ell=0}\Phi_\ell e_{t-\ell},
\qquad
\text{for}\ t=t^*,\ldots,T,
\label{eq:linear-process}
\end{equation}
where dimension $(p)$ vector $\varepsilon_t$ denote the innovation; $\{e_t\}$ is a sequence of i.i.d. random vector with dimension $(p)$, $\mathbb E(e_t)=0$, $\mathbb E(e_t e_t')=V_e>0$; $\Phi(L)=\sum^\infty_{\ell=0}\Phi_\ell L^\ell$ is a sequence of $(p)\times(p)$ matrices, $L$ is the lag operator, and
\begin{equation}
\sum_{\ell=0}^{\infty}\ell\|\Phi_\ell\|<\infty.
\label{eq:summability}
\end{equation}
$\|\cdot \|$ denotes the Euclidean norm of a vector or the Frobenius norm of a matrix. Then
\begin{equation}
N^{-1/2}\sum_{t=t^*}^{\lfloor Ts\rfloor}\varepsilon_t
\Rightarrow
W(s),
\qquad
s\in[t^*/T,1],
\label{eq:fclt}
\end{equation}
where $W(\cdot)$ is a Brownian motion with long-run covariance matrix $\Omega_\varepsilon$.\\

Moreover, the integrated level component of $X_t$ is linked to the permanent innovation:
\begin{equation}
N^{-1/2}X_{\lfloor Ts\rfloor}
\Rightarrow
W(s),
\qquad
s\in[t^*/T,1].
\end{equation}

Define the realised local Brownian direction
\begin{equation}
b_{\delta_0}:=W \bigg(\frac{\delta_0-\delta^\ast}{1-\delta^\ast}\bigg),
\qquad
\mathcal{Q}_{\delta_0}:=\frac{b_{\delta_0}}{\|b_{\delta_0}\|},
\label{eq:brownian-direction}
\end{equation}
and let $\mathcal{Q}_{\delta_0}^{\perp}$ be an orthogonal complement, so that
\begin{equation}
\mathcal{Q}_{\delta_0}^{x}:=\big(\mathcal{Q}_{\delta_0},\mathcal{Q}_{\delta_0}^{\perp}\big)
\label{eq:regressor-space-rotation-matrix}
\end{equation}
is a $p\times p$ orthogonal matrix.\\

Now define the local weighted level signal in regressor space by
\begin{equation}
G_N^{x}(\delta_0)
:=
\sum_{t=t^*+1}^{T}
K_{t,h}(\delta_0)\,
R_{1t}^\circ(\delta_0)R_{1t}^\circ(\delta_0)'.
\label{eq:regressor-space-signal}
\end{equation}
By the same kernel-degeneracy argument as in Phillips et al.\ (2017), the leading unrotated term of $G_N^{x}(\delta_0)$ is asymptotically proportional to
\begin{equation}
b_{\delta_0}b_{\delta_0}',
\label{eq:rank-one-regressor-limit}
\end{equation}
so the raw local signal is asymptotically rank deficient in un-rotated coordinates.\\

Let $\widehat {\mathcal Q}_{\delta_0}^{x}$ be a sample analogue of $\mathcal{Q}_{\delta_0}^{x}$, and define the regressor-space scaling matrix
\begin{equation}
\mathcal D_{N,\delta_0}^{x}
=
\operatorname{diag}\{N\sqrt h,\;Nh\,I_{p-1}\}.
\label{eq:regressor-space-scaling-matrix}
\end{equation}
Assume that
\begin{equation}
(\mathcal D_{N,\delta_0}^{x})^{-1}
\widehat {\mathcal Q}_{\delta_0}^{x\prime}
G_N^{x}(\delta_0)
\widehat {\mathcal Q}_{\delta_0}^{x}
(\mathcal D_{N,\delta_0}^{x})^{-1}
\Rightarrow
\Delta_{\delta_0}^{x},
\label{eq:regressor-space-rotated-limit}
\end{equation}
where $\Delta_{\delta_0}^{x}$ is finite and nonsingular with probability one.\\

Let $\widehat{\mathcal L}_{\delta_0}(\vartheta)$ denote the concentrated local profile criterion obtained from the kernel-WLS criterion after substituting out $\alpha(\delta_0)$. Define the parameter-space score and Hessian by
\begin{equation}
\mathcal S_N^{\vartheta}(\delta_0)
:=
\frac{\partial \widehat{\mathcal L}_{\delta_0}(\vartheta)}
{\partial \vartheta}
\Bigg|_{\vartheta=\vartheta(\delta_0)}
\in\mathbb R^{m},
\label{eq:parameter-space-score}
\end{equation}
and
\begin{equation}
\mathcal H_N^{\vartheta}(\delta_0)
:=
\frac{\partial^2 \widehat{\mathcal L}_{\delta_0}(\vartheta)}
{\partial \vartheta\partial \vartheta'}
\Bigg|_{\vartheta=\vartheta(\delta_0)}
\in\mathbb R^{m\times m}.
\label{eq:parameter-space-hessian}
\end{equation}
The dimension $m=r(p-r)$ here is the number of free parameters in the normalised cointegrating vector. \\

Since $\vartheta(\delta_0)$ parameterises the free block of $\beta(\delta_0)$, the parameter-space Hessian is induced by the regressor-space signal through the normalisation map. Let
\begin{equation}
J_{\delta_0}
:=
\frac{\partial\operatorname{vec}\{\beta(\delta_0)\}}
{\partial\vartheta(\delta_0)'}
\in\mathbb R^{pr\times m}.
\end{equation}
Thus, for a local perturbation $d\vartheta$,
\begin{equation}
d\operatorname{vec}\{\beta(\delta_0)\}
=
J_{\delta_0}d\vartheta.
\end{equation}
Let
$U_N^{x\beta}(\delta_0)\in\mathbb R^{p\times r}$ denote the corresponding raw local score matrix in the unrestricted $\operatorname{vec}(\beta)$-space. Then the parameter-space Hessian and score are induced by the unrestricted objects through the normalisation map:
\begin{equation}
\mathcal H_N^{\vartheta}(\delta_0)
=
J_{\delta_0}'
\big(I_r\otimes G_N^{x}(\delta_0)\big)
J_{\delta_0}
+
o_p(N^2h),
\label{eq:parameter-hessian-induced}
\end{equation}
and
\begin{equation}
\mathcal S_N^{\vartheta}(\delta_0)
=
J_{\delta_0}'
\operatorname{vec}\!\big(U_N^{x\beta}(\delta_0)\big)
+
o_p(N\sqrt h).
\label{eq:parameter-score-induced}
\end{equation}
Here,
\begin{equation}
I_r\otimes G_N^{x}(\delta_0)\in\mathbb R^{pr\times pr},
\qquad
\operatorname{vec}\{U_N^{x\beta}(\delta_0)\}\in\mathbb R^{pr}.
\end{equation}

Since $C(\delta_0)$ has $r$ columns, the parameter space inherits $r$ fast directions and $m-r$ slow directions. Let
\begin{equation}
\mathcal{Q}_{\delta_0}^{\vartheta}
=
\big(
\mathcal{Q}_{\delta_0}^{\vartheta,f},
\mathcal{Q}_{\delta_0}^{\vartheta,s}
\big)
\label{eq:parameter-space-rotation}
\end{equation}
be an $m\times m$ orthogonal matrix, where $\mathcal{Q}_{\delta_0}^{\vartheta,f}$ is $m\times r$ and $\mathcal{Q}_{\delta_0}^{\vartheta,s}$ is $m\times(m-r)$, and define
\begin{equation}
\mathcal D_{N,\delta_0}^{\vartheta}
=
\operatorname{diag}\{N\sqrt h\,I_r,\;Nh\,I_{m-r}\}.
\label{eq:parameter-space-scaling}
\end{equation}
Assume further that
\begin{equation}
(\mathcal D_{N,\delta_0}^{\vartheta})^{-1}
\mathcal{Q}_{\delta_0}^{\vartheta\prime}
\mathcal H_N^{\vartheta}(\delta_0)
\mathcal{Q}_{\delta_0}^{\vartheta}
(\mathcal D_{N,\delta_0}^{\vartheta})^{-1}
\Rightarrow
\Xi_{\delta_0},
\label{eq:parameter-hessian-limit}
\end{equation}
and
\begin{equation}
(\mathcal D_{N,\delta_0}^{\vartheta})^{-1}
\mathcal{Q}_{\delta_0}^{\vartheta\prime}
\mathcal S_N^{\vartheta}(\delta_0)
\Rightarrow
\Upsilon_{\delta_0},
\label{eq:parameter-score-limit}
\end{equation}
where $\Xi_{\delta_0}$ and $\Upsilon_{\delta_0}$ are Brownian functionals generated by the post-intervention Brownian motion $W(s)$. In particular, $\Xi_{\delta_0}$ is induced by the rotated local level signal, while $\Upsilon_{\delta_0}$ is induced by the local score involving the product of the level process and the innovation process.\\

To derive the asymptotic behaviour of $\widehat\vartheta(\delta_0)$, we impose an assumption of high-level local quadratic approximation for the concentrated criterion. Such an expansion is standard in the general extremum-estimation and $M$-estimation literature, where local asymptotic analysis proceeds by representing the criterion through a linear score term, a quadratic curvature term, and a remainder that is asymptotically negligible on compact sets (see \shortciteA{newey1994large}, \shortciteA{andrews1994asymptotics}). In the present setting, however, a primitive verification is non-trivial because the criterion depends on kernel-weighted local moments, generated residuals, and integrated regressors. In particular, for multivariate time-varying cointegrating models with kernel estimation, the usual asymptotic methods may break down and require rotation and rescaling arguments \shortcite{phillips2017estimating}.

\paragraph{Assumption 3.6 (Local quadratic approximation of the concentrated criterion).}
For each fixed $c>0$,
\begin{align}
&\sup_{\|z\|\le c}
\bigg|
\widehat{\mathcal L}_{\delta_0}
\big(
\vartheta(\delta_0)+
\mathcal Q_{\delta_0}^{\vartheta}
(\mathcal D_{N,\delta_0}^{\vartheta})^{-1}z
\big)
-
\widehat{\mathcal L}_{\delta_0}
\big(\vartheta(\delta_0)\big)
\notag\\
&\quad
+
z'
(\mathcal D_{N,\delta_0}^{\vartheta})^{-1}
\mathcal Q_{\delta_0}^{\vartheta\prime}
\mathcal S_N^{\vartheta}(\delta_0)
\notag\\
&\quad
+
\frac{1}{2}
z'
(\mathcal D_{N,\delta_0}^{\vartheta})^{-1}
\mathcal Q_{\delta_0}^{\vartheta\prime}
\mathcal H_N^{\vartheta}(\delta_0)
\mathcal Q_{\delta_0}^{\vartheta}
(\mathcal D_{N,\delta_0}^{\vartheta})^{-1}
z
\bigg|
=o_p(1).
\label{eq:scaled-local-expansion}
\end{align}
\noindent
Assumption 3.6 requires that, after the appropriate scaling, the concentrated local criterion admits a uniform second-order expansion in a small neighbourhood of the true parameter $\vartheta(\delta_0)$. The score vector $\mathcal S_N^{\vartheta}(\delta_0)$ and the Hessian matrix $\mathcal H_N^{\vartheta}(\delta_0)$ play the roles of the local score and local curvature, respectively, while the remainder term is asymptotically negligible uniformly over compact sets in the local coordinates.

\begin{theorem}[Asymptotic convergence of $\widehat\beta(\delta_0)$]
\label{thm:beta-convergence}
Suppose Assumptions 3.1--3.6 hold, together with local identification,
\eqref{eq:parameter-hessian-limit}--\eqref{eq:scaled-local-expansion}. Then
\begin{equation}
\mathcal D_{N,\delta_0}^{\vartheta}
\mathcal{Q}_{\delta_0}^{\vartheta\prime}
\big(
\widehat\vartheta(\delta_0)-\vartheta(\delta_0)
\big)
=O_p(1).
\label{eq:scaled-vartheta}
\end{equation}

Consequently, the fast block satisfies
\begin{equation}
\mathcal{Q}_{\delta_0}^{\vartheta,f\prime}
\big(
\widehat\vartheta(\delta_0)-\vartheta(\delta_0)
\big)
=
O_p\!\big((N\sqrt h)^{-1}\big),
\label{eq:beta-fast-rate}
\end{equation}
whereas the slow block satisfies
\begin{equation}
\mathcal{Q}_{\delta_0}^{\vartheta,s\prime}
\big(
\widehat\vartheta(\delta_0)-\vartheta(\delta_0)
\big)
=
O_p\!\big((Nh)^{-1}\big).
\label{eq:beta-slow-rate}
\end{equation}
Since the map $C\mapsto\beta$ is linear under \eqref{eq:beta-normalisation}, the same directional rates apply to the free block of $\widehat\beta(\delta_0)-\beta(\delta_0)$.
\end{theorem}
The proof is in the Appendix A.\\

The decomposition into fast and slow parameter directions is stated for the general $p$-dimensional, rank-$r$ case. In the bivariate rank-one case, $m=r(p-r)=1$, so the parameter-space scaling reduces to the single fast rate $N\sqrt h$.\\

For the adjustment coefficient matrix, $\widehat\beta(\delta_0)$ is super-consistent, so its estimation error is asymptotically negligible in the first-order expansion of $\widehat\alpha(\delta_0)$. Define
\begin{equation}
\Psi(S_{01},S_{11},\beta)
=
S_{01}\beta\big(\beta'S_{11}\beta\big)^{-1}.
\label{eq:alpha-map}
\end{equation}
Then
\begin{equation}
\widehat\alpha(\delta_0)
=
\Psi\big(
\widehat S_{01}(\delta_0),
\widehat S_{11}(\delta_0),
\widehat\beta(\delta_0)
\big).
\label{eq:alpha-estimator-map}
\end{equation}

Assume that there exist finite matrices $S_{01}(\delta_0)$ and
\begin{equation}
M_{\delta_0}
:=
\beta(\delta_0)'S_{11}(\delta_0)\beta(\delta_0),
\label{eq:cointegrating-s11-limit}
\end{equation}
with $M_{\delta_0}$ nonsingular, such that the stationary local moment vector entering the first-order expansion of $\widehat\alpha(\delta_0)$ obeys a root-$Nh$ central limit theorem.

\begin{theorem}[Asymptotic convergence of $\widehat\alpha(\delta_0)$]
\label{thm:alpha-convergence}
Suppose the conditions of Theorem \ref{thm:beta-convergence} hold. In addition, suppose that the stationary local moment vector underlying the map $\Psi$ satisfies a root-$Nh$ central limit theorem. Then
\begin{equation}
\widehat\alpha(\delta_0)-\alpha(\delta_0)
=
\mathcal A_{N,1}(\delta_0)
+
\mathcal A_{N,2}(\delta_0)
+
o_p\!\big((Nh)^{-1/2}\big),
\label{eq:alpha-asymptotic-linear-representation}
\end{equation}
where $\mathcal A_{N,1}(\delta_0)$ is linear in
$\widehat S_{01}(\delta_0)-S_{01}(\delta_0)$ and $\mathcal A_{N,2}(\delta_0)$ is linear in the stationary component of
$\widehat S_{11}(\delta_0)-S_{11}(\delta_0)$.

In particular,
\begin{equation}
\sqrt{Nh}\,
\operatorname{vec}\big(
\widehat\alpha(\delta_0)-\alpha(\delta_0)
\big)
=O_p(1).
\label{eq:alpha-root-Nh-rate}
\end{equation}
\end{theorem}
The proof is in the Appendix A.

\subsection{Asymptotic distribution of the local estimators}
\label{subsec:asymptotic-distribution-estimators}

This subsection strengthens the convergence-rate results above into distributional statements. The consistency result in Theorem \ref{thm:consistency} ensures that the local estimators are centred around the true post-intervention coefficients at the fixed interior point $\delta_0$, while Theorems \ref{thm:beta-convergence} and \ref{thm:alpha-convergence} provide the relevant convergence rates. The remaining step is to characterise the limiting law of the scaled estimation errors. The following result is stated in a high-level form because the explicit Brownian functional depends on the residualisation step and on the local rotation used to remove the kernel-induced degeneracy.\\

The distribution theory is nonstandard for $\widehat\beta(\delta_0)$. The reason is that the local reduced-rank problem contains the integrated regressor $X_{t-k}$. Under Assumption 3.5,
\[
N^{-1/2}X_{\lfloor Ts\rfloor}
\Rightarrow 
W(s), \qquad s \in[t^*/T,1]
\]
so the local signal matrix is driven by the Brownian path of the integrated component. Hence the limiting distribution of the local cointegrating vector is generally a Brownian functional rather than a conventional Gaussian distribution. This is the same source of nonstandard behaviour that arises in kernel estimation of smooth time-varying cointegration models, where the local weighted signal matrix generated by nonstationary regressors becomes asymptotically degenerate in unrotated coordinates.\\

Recall that
\[
b_{\delta_0}=W\bigg(\frac{\delta_0-\delta^\ast}{1-\delta^\ast}\bigg),
\qquad
\mathcal{Q}_{\delta_0}=\frac{b_{\delta_0}}{\|b_{\delta_0}\|},
\]
and that $\mathcal{Q}_{\delta_0}^{x}$ rotates the regressor space into the Brownian direction $\mathcal{Q}_{\delta_0}$ and its orthogonal complement. The corresponding parameter-space rotation and scaling are given by $\mathcal{Q}_{\delta_0}^{\vartheta}$ and $\mathcal D_{N,\delta_0}^{\vartheta}$, as defined in \eqref{eq:parameter-space-rotation} and \eqref{eq:parameter-space-scaling}.\\

For the distributional result, strengthen \eqref{eq:parameter-hessian-limit} and \eqref{eq:parameter-score-limit} to joint weak convergence:
\begin{equation}
\left(
(\mathcal D_{N,\delta_0}^{\vartheta})^{-1}
\mathcal{Q}_{\delta_0}^{\vartheta\prime}
\mathcal H_N^{\vartheta}(\delta_0)
\mathcal{Q}_{\delta_0}^{\vartheta}
(\mathcal D_{N,\delta_0}^{\vartheta})^{-1},
\,
(\mathcal D_{N,\delta_0}^{\vartheta})^{-1}
\mathcal{Q}_{\delta_0}^{\vartheta\prime}
\mathcal S_N^{\vartheta}(\delta_0)
\right)
\Rightarrow
\left(
\Xi_{\delta_0},
\Upsilon_{\delta_0}
\right),
\label{eq:joint-score-hessian-limit}
\end{equation}
where $\Xi_{\delta_0}$ is nonsingular with probability one.\\

More explicitly, the Hessian limit can be represented at a high level as
\begin{equation}
\Xi_{\delta_0}
=
\mathcal J_{\delta_0}'
\left[
I_r\otimes
\Delta_{\delta_0}^{x}
\right]
\mathcal J_{\delta_0},
\label{eq:Xi-brownian-functional}
\end{equation}
where $\mathcal J_{\delta_0}$ denotes the Jacobian induced by the normalisation map
\[
\vartheta(\delta_0)\mapsto \operatorname{vec}\big(\beta(\delta_0)\big)
\]
after rotation into the parameter space, and
\[
\Delta_{\delta_0}^{x}
=
\int K(s)\,
\mathcal X_{\delta_0}(s)\mathcal X_{\delta_0}(s)'\,ds
\]
where $\mathcal X_{\delta_0}(s)$ denotes the limiting rotated local regressor process after the fast-slow rescaling in $\mathcal D_{N,\delta_0}^{x}$. Its exact expression depends on the particular residualisation of $X_{t-k}$ with respect to $Z_t$.\\

Similarly, the score limit can be written as
\begin{equation}
\Upsilon_{\delta_0}
=
\mathcal J_{\delta_0}'
\operatorname{vec}
\left[
\int K(s)\,
\mathcal X_{\delta_0}(s)
\,dW_{\delta_0}(s)'
\right],
\label{eq:Upsilon-brownian-functional}
\end{equation}
where $W_{\delta_0}$ denotes the Brownian innovation component for fixed $\delta_0$.

\begin{theorem}[Asymptotic distribution of $\widehat\beta(\delta_0)$]
\label{thm:beta-distribution}
Suppose the conditions of Theorem \ref{thm:beta-convergence} hold. Suppose further that the joint score-Hessian convergence in \eqref{eq:joint-score-hessian-limit} holds and that the limiting quadratic criterion has a unique minimiser. Then
\begin{equation}
\mathcal D_{N,\delta_0}^{\vartheta}
\mathcal{Q}_{\delta_0}^{\vartheta\prime}
\big(
\widehat\vartheta(\delta_0)-\vartheta(\delta_0)
\big)
\Rightarrow
-\Xi_{\delta_0}^{-1}\Upsilon_{\delta_0}.
\label{eq:theta-distribution}
\end{equation}
Equivalently,
\begin{equation}
\mathcal D_{N,\delta_0}^{\vartheta}
\mathcal{Q}_{\delta_0}^{\vartheta\prime}
\big(
\widehat\vartheta(\delta_0)-\vartheta(\delta_0)
\big)
\Rightarrow
-\bigg[ \mathcal J_{\delta_0}'
\left[
I_r\otimes
\Delta_{\delta_0}^{x}
\right]
\mathcal J_{\delta_0}  \bigg]^{-1}
\mathcal J_{\delta_0}'
\operatorname{vec}
\left[
\int K(s)\,
\mathcal X_{\delta_0}(s)
\,dW_{\delta_0}(s)'
\right].
\label{eq:theta-brownian-functional-limit}
\end{equation}
Since $\vartheta(\delta_0)=\operatorname{vec}(C(\delta_0))$ and
\[
\beta(\delta_0)
=
\begin{pmatrix}
I_r\\
C(\delta_0)
\end{pmatrix},
\]
the same Brownian-functional limit applies to the lower block of
$\widehat\beta(\delta_0)-\beta(\delta_0)$.
\end{theorem}

The proof is in the Appendix A.\\

Theorem \ref{thm:beta-distribution} shows that the local cointegrating-vector estimator has a mixed, generally non-Gaussian, limit distribution. The randomness enters through two channels. First, the local Hessian depends on the Brownian level signal of the integrated regressor. Second, the local score depends on the stochastic integral between the local Brownian regressor signal and the innovation Brownian motion. This is why the convergence result in Theorem \ref{thm:beta-convergence} gives rates, while the distributional result requires explicitly retaining the Brownian-functional form.\\

Next consider the adjustment matrix. Recall that
\[
\widehat\alpha(\delta_0)
=
\Psi
\big(
\widehat S_{01}(\delta_0),
\widehat S_{11}(\delta_0),
\widehat\beta(\delta_0)
\big),
\qquad
\Psi(S_{01},S_{11},\beta)
=
S_{01}\beta(\beta'S_{11}\beta)^{-1}.
\]
Although $\widehat\beta(\delta_0)$ has a nonstandard Brownian-functional limit, its estimation error is of order $O_p((Nh)^{-1})$ in the slowest direction by Theorem \ref{thm:beta-convergence}. This is smaller than the root-$Nh$ scale relevant for $\widehat\alpha(\delta_0)$. Therefore the first-order distribution of $\widehat\alpha(\delta_0)$ is driven by the stationary local moments entering $\widehat S_{01}(\delta_0)$ and the cointegrating component of $\widehat S_{11}(\delta_0)$.\\

Assume that the stationary local moment vector in Theorem \ref{thm:alpha-convergence} satisfies the central limit theorem
\begin{equation}
\sqrt{Nh}
\begin{pmatrix}
\operatorname{vec}
\big(
\widehat S_{01}(\delta_0)-S_{01}(\delta_0)
\big)
\\[0.3em]
\operatorname{vec}
\big(
\widehat S_{11}^{c}(\delta_0)-S_{11}^{c}(\delta_0)
\big)
\end{pmatrix}
\Rightarrow
\mathcal N(0,\Omega_{\alpha,\delta_0}),
\label{eq:alpha-stationary-moment-clt}
\end{equation}
where $\widehat S_{11}^{c}(\delta_0)$ denotes the stationary cointegrating component of $\widehat S_{11}(\delta_0)$ entering
\[
\widehat\beta(\delta_0)'
\widehat S_{11}(\delta_0)
\widehat\beta(\delta_0).
\]

\begin{theorem}[asymptotic distribution of $\widehat\alpha(\delta_0)$]
\label{thm:alpha-distribution}
Suppose the conditions of Theorem \ref{thm:alpha-convergence} hold. Suppose further that the stationary local moment CLT in \eqref{eq:alpha-stationary-moment-clt} holds and that the local smoothing bias is negligible at the root-$Nh$ scale:
\begin{equation}
\sqrt{Nh}\,h^\gamma\to 0.
\label{eq:bias-negligible-alpha}
\end{equation}
Then
\begin{equation}
\sqrt{Nh}\,
\operatorname{vec}
\big(
\widehat\alpha(\delta_0)-\alpha(\delta_0)
\big)
\Rightarrow
\mathcal N
\big(
0,
\mathcal V_{\alpha,\delta_0}
\big),
\label{eq:alpha-distribution}
\end{equation}
where
\begin{equation}
\mathcal V_{\alpha,\delta_0}
=
\mathcal W_{\alpha,\delta_0}
\Omega_{\alpha,\delta_0}
 \mathcal W_{\alpha,\delta_0}',
\label{eq:alpha-asymptotic-variance}
\end{equation}
$\mathcal W_{\alpha,\delta_0}$ denotes the derivative matrix of the map $\Psi$ with respect to stationary local moments, evaluated at
\[
\big(
S_{01}(\delta_0),
S_{11}(\delta_0),
\beta(\delta_0)
\big).
\]
\end{theorem}
The proof is in the Appendix A.

\section{Testing algorithm for the estimators}
\label{Testing algorithm3}

This section develops two bootstrap-based tests for post-intervention structural change. The first test examines whether the post-intervention coefficient functions remain equal to the pre-intervention baseline. The second test examines whether the post-intervention long-run matrix follows a smooth path, rather than displaying an unknown level shift or kink. The first test is related
to time-varying cointegration and parameter-constancy testing in cointegrated VAR systems; see \shortciteA{hansen1999some}, \shortciteA{bierens2010time} and \shortciteA{martins2018bootstrap}. The second test is related to the literature on smooth structural change and time-varying parameter testing; see \shortciteA{andrews1993tests}, \shortciteA{lin1994testing} and
\shortciteA{elliott2006efficient}, as well as fit-loss type specification testing; see \shortciteA{chen2012testing} and \shortciteA{hong2013loss}.

\subsection{Test for post-intervention baseline constancy}
\label{subsec:test1_algorithm}

The purpose of the first test is to assess whether the post-intervention dynamics remain stable relative to the long-run cointegrating relation identified before the intervention. The test compares two local descriptions of the post-intervention system at each evaluation point. The unrestricted local fit allows the long-run matrix to vary over time, whereas the baseline-restricted fit imposes the pre-intervention long-run benchmark on the post-intervention observations. If the intervention does not alter the long-run equilibrium structure, imposing the baseline restriction should not generate a large loss of local fit. A persistent increase in the local fit loss therefore provides evidence that the post-intervention system departs from the pre-intervention benchmark.\\

Under the normalisation
\[
\beta(\delta)=
\begin{pmatrix}
I_r\\
C(\delta)
\end{pmatrix},
\]
the null and alternative hypotheses are
\begin{equation}
H^{\text{Base}}_0:\ 
\alpha(\delta)=\alpha_0,
\quad
\beta(\delta)=\beta_0,
\qquad
\forall \delta\in[\delta^\ast,1],
\label{eq:test1_null}
\end{equation}
against
\begin{equation}
H^{\text{Base}}_1:\ 
\exists\,\delta\in[\delta^\ast,1]
\ \text{such that}\
\big(\alpha(\delta),\beta(\delta)\big)
\neq
\big(\alpha_0,\beta_0\big).
\label{eq:test1_alt}
\end{equation}

Let
\[
\mathbb G_N=\{g_1,\ldots,g_M\}\subset(\delta^\ast,1)
\]
denote an interior evaluation grid on the post-intervention interval. For each
$g_m\in\mathbb G_N$, define the unrestricted local minimum
\begin{equation}
 Q_{U}(g_m)
:=
\min_{\alpha,C}
\mathcal L_{g_m}
\!\left(
\alpha(g_m),
\begin{pmatrix}
I_r\\
C(g_m)
\end{pmatrix}
\right),
\label{eq:test1_qU}
\end{equation}
where $\mathcal L_{g_m}(\cdot,\cdot)$ is the local reduced-rank criterion
defined in \eqref{eq:criterion}. Let
\begin{equation}
\widehat\beta_0
=
\begin{pmatrix}
I_r\\
\widehat C_0
\end{pmatrix}
\label{eq:test1_pre_beta}
\end{equation}
be the pre-intervention Johansen estimator under the same normalisation. The
baseline-restricted local criterion is
\begin{equation}
 Q_{\text{Base}}(g_m)
:=
\mathcal L_{g_m}
\!\left(
\widehat\alpha_0,
\widehat\beta_0
\right).
\label{eq:test1_qB}
\end{equation}
The local baseline-deviation score is
\begin{equation}
 D_{\text{Base}}(g_m)
:=
\widehat N_h(g_m)
\Big |
 Q_{\text{Base}}(g_m)- Q_{U}(g_m)
\Big |,
\label{eq:test1_local_score}
\end{equation}
where
\[
\widehat N_h(g_m)
=
\sum^T_{t=t^*+1}K_{t,h}(g_m)
\]
is the local kernel normalising factor. The global statistic is
\begin{equation}
\mathcal Z_{\text{Base},N}
:=
\frac{1}{M}\sum_{m=1}^{M} D_{\text{Base}}(g_m).
\label{eq:test1_statistic}
\end{equation}
Large values of $\mathcal Z_{\text{Base},N}$ provide evide
nce against $H^{\text{Base}}_0$.\\

The procedure of Test 1 step by step: 
First estimate $(\widehat\alpha_0,\widehat\beta_0)$ from the
pre-intervention subsample $t=1,\ldots,t^*$ by Johansen's maximum-likelihood method, imposing the normalisation in \eqref{eq:test1_pre_beta}. Choose $\mathbb G_N$ and the first-stage bandwidth $h$. For each $g_m\in\mathbb G_N$, compute the local reduced-rank estimator
$(\widehat\alpha(g_m),\widehat\beta(g_m))$ and record $Q_U(g_m)$. Evaluate the same criterion at $(\widehat\alpha_0,\widehat\beta_0)$ to obtain $Q_{\text{Base}}(g_m)$. Then compute $D_{\text{Base}}(g_m)$ and aggregate the local scores to obtain $\mathcal Z_{\text{Base},N}$. The null is rejected at nominal level $\tau$ if
\[
\mathcal Z_{\text{Base},N}>c_{\text{Base},1-\tau}^{\ast},
\]
where $c_{\text{Base},1-\tau}^{\ast}$ is obtained from Algorithm \ref{alg:test1_bootstrap} in Appendix B.\\

The statistic $\mathcal Z_{\text{Base},N}$ compares the unrestricted local post-intervention fit with the fit obtained by imposing the pre-intervention baseline throughout the post-intervention period. It is therefore designed to detect post-intervention structural change from the pre-intervention long-run benchmark.

\subsection{Test for smooth post-intervention evolution of the long-run matrix}
\label{subsec:test2_algorithm}

Test 2 focuses on the internal regularity of the post-intervention long-run
matrix path, rather than on its distance from the pre-intervention benchmark.
Specifically, it asks whether the estimated path of
\[
\Pi(\delta)=\alpha(\delta)\beta(\delta)'
\]
can be interpreted as a smoothly evolving matrix-valued function over the
post-intervention interval. The null hypothesis is formulated in terms of
$C^1$-smoothness. For a functional matrix
$\Pi:[\delta^\ast,1]\to\mathbb R^{p\times p}$, the condition
$\Pi(\cdot)\in C^1([\delta^\ast,1])$ means that $\Pi(\delta)$ is
continuously differentiable on $[\delta^\ast,1]$. Equivalently, both
$\Pi(\delta)$ and
\[
\dot \Pi(\delta)
=
\frac{d\Pi(\delta)}{d\delta}
\]
are continuous over the post-intervention interval. The test is therefore designed to detect whether the post-intervention long-run matrix evolves as a continuously differentiable path or instead contains an unknown jump, kink, or other non-$C^1$ feature.\\

The null and alternative hypotheses are
\begin{equation}
H^{\text{smo}}_0:\ 
\Pi(\cdot)\in C^1([\delta^\ast,1];\mathbb R^{p\times p}),
\label{eq:test2_null}
\end{equation}
against
\begin{equation}
H^{\text{smo}}_1:\ 
\Pi(\cdot)\notin C^1([\delta^\ast,1];\mathbb R^{p\times p}).
\label{eq:test2_alt}
\end{equation}
Thus, the alternative allows for interior level shifts or slope kinks.\\

Let
\[
\mathbb G_N=\{g_1,\ldots,g_M\}\subset(\delta^\ast,1)
\]
denote the first-stage post-intervention evaluation grid. For each
$g_m\in\mathbb G_N$, define
\begin{equation}
\widehat\Pi(g_m)
=
\widehat\alpha(g_m)\widehat\beta(g_m)',
\qquad
m=1,\ldots,M,
\label{eq:test2_pi_hat_grid}
\end{equation}
where $\widehat\alpha(g_m)$ and $\widehat\beta(g_m)$ are the local
reduced-rank estimators from Section \ref{LST3}. Let
\begin{equation}
\widehat\zeta(g_m)
:=
\operatorname{vec}\!\big(\widehat\Pi(g_m)\big)
\in\mathbb R^{p^2},
\qquad
m=1,\ldots,M.
\label{eq:test2_zeta}
\end{equation}

Let $h_{\text{smo}}$ denote the second-stage bandwidth and letn$K_{\text{smo}}(\cdot)$ be a second-stage kernel. For each interior candidate point $g_j\in\mathbb G_N$, define the pooled weights
\begin{equation}
K_{\text{smo},m}(g_j)
=
K_{\text{smo}}\!\left(\frac{g_m-g_j}{h_{\text{smo}}}\right),
\qquad
m=1,\ldots,M,
\label{eq:test2_weight_pooled}
\end{equation}
together with the left- and right-sided weights
\begin{equation}
K_{\text{smo},m}^{-}(g_j)
=
K_{\text{smo},m}(g_j)\mathbf 1\{g_m<g_j\},
\label{eq:test2_weight_left}
\end{equation}
and
\begin{equation}
K_{\text{smo},m}^{+}(g_j)
=
K_{\text{smo},m}(g_j)\mathbf 1\{g_m\ge g_j\}.
\label{eq:test2_weight_right}
\end{equation}

Under $H^{\text{smo}}_0$, the path $\widehat\zeta(\cdot)$ should admit a common local-linear approximation in a neighbourhood of $g_j$. The pooled local-linear criterion is
\begin{equation}
 Q_{\text{smo},P}(g_j)
:=
\min_{a_{0,j},a_{1,j}}
\sum_{m=1}^{M}
K_{\text{smo},m}(g_j)
\left\|
\widehat\zeta(g_m)-a_{0,j}-a_{1,j}(g_m-g_j)
\right\|^2,
\label{eq:test2_qP}
\end{equation}
where $a_{0,j},a_{1,j}\in\mathbb R^{p^2}$. To allow for a failure of $C^1$-smoothness, define the split local-linear criterion
\begin{align}
 Q_{\text{smo},S}(g_j)
:=
\min_{a^{-}_{0,j},a^{-}_{1,j},a^{+}_{0,j},a^{+}_{1,j}}
\Bigg[
&
\sum_{m=1}^{M}
K_{\text{smo},m}^{-}(g_j)
\left\|
\widehat\zeta(g_m)-a^{-}_{0,j}-a^{-}_{1,j}(g_m-g_j)
\right\|^2
\notag\\
&+
\sum_{m=1}^{M}
K_{\text{smo},m}^{+}(g_j)
\left\|
\widehat\zeta(g_m)-a^{+}_{0,j}-a^{+}_{1,j}(g_m-g_j)
\right\|^2
\Bigg],
\label{eq:test2_qS}
\end{align}
where
\[
a^{-}_{0,j},a^{+}_{0,j}\in\mathbb R^{p^2},
\qquad
a^{-}_{1,j},a^{+}_{1,j}\in\mathbb R^{p^2}.
\]
The split fit allows both the local intercept and the local slope to differ across the two sides of $g_j$. Hence,
\[
 Q_{\text{smo},P}(g_j)- Q_{\text{smo},S}(g_j)
\]
is informative about either a jump in $\Pi(\cdot)$ or a kink in its local slope.\\

Define the second-stage local kernel normalising factors
\begin{equation}
\widehat N_{\text{smo}}(g_j)
:=
\sum_{m=1}^{M}K_{\text{smo},m}(g_j),
\label{eq:test2_effective_mass}
\end{equation}
and
\begin{equation}
\widehat N_{\text{smo}}^{-}(g_j)
:=
\sum_{m=1}^{M}K_{\text{smo},m}^{-}(g_j),
\qquad
\widehat N_{\text{smo}}^{+}(g_j)
:=
\sum_{m=1}^{M}K_{\text{smo},m}^{+}(g_j).
\label{eq:test2_effective_mass_sides}
\end{equation}
For each admissible $g_j$, define the local smoothness-violation score
\begin{equation}
 D_{\text{smo}}(g_j)
:=
\widehat N_{\text{smo}}(g_j)
\Big|
 Q_{\text{smo},P}(g_j)- Q_{\text{smo},S}(g_j)
\Big|.
\label{eq:test2_local_score}
\end{equation}
The test statistic is
\begin{equation}
\mathcal Z_{\text{smo},N}
:=
\max_{g_j\in\mathbb G_{N,\text{adm}}}
 D_{\text{smo}}(g_j).
\label{eq:test2_statistic}
\end{equation}
Large values of $\mathcal Z_{\text{smo},N}$ provide evidence against
$H^{\text{smo}}_0$.\\

Test 2 procedure step by step: Choose the first-stage post-intervention grid $\mathbb G_N$, the second-stage bandwidth $h_{\text{smo}}$, the second-stage kernel $K_{\text{smo}}(\cdot)$, and the trimming threshold $\nu_N$. For each $g_m\in\mathbb G_N$, compute the first-stage local estimators $\widehat\alpha(g_m)$ and $\widehat\beta(g_m)$, and form $\widehat\Pi(g_m)$ and $\widehat\zeta(g_m)$. For each $g_j\in\mathbb G_N$, compute $Q_{\text{smo},P}(g_j)$, $Q_{\text{smo},S}(g_j)$, $\widehat N_{\text{smo}}^{-}(g_j)$, and $\widehat N_{\text{smo}}^{+}(g_j)$. Retain only $g_j\in\mathbb G_{N,\text{adm}}$, compute $D_{\text{smo}}(g_j)$, and form
$\mathcal Z_{\text{smo},N}$. The null is rejected at nominal level $\tau$ if
\[
\mathcal Z_{\text{smo},N}>c_{\text{smo},1-\tau}^{\ast},
\]
where $c_{\text{smo},1-\tau}^{\ast}$ is obtained from Algorithm \ref{alg:test2_bootstrap} in A
ppendix B.\\

The statistic $\mathcal Z_{\text{smo},N}$ compares a pooled local-linear approximation with a split left-right local-linear approximation applied to the estimated post-intervention path $\widehat\Pi(\cdot)$. It is therefore tailored to detect failures of $C^1$-smoothness in the post-intervention evolution of the long-run matrix, including both abrupt level shifts and slope discontinuities.

\section{Simulation study}
\label{sec:simulation3}

In this section, we use simulated data process to evaluate the perfomance of two tests.

\subsection{Simulation design and parameter setting}
\label{subsec:simulation_design}

The data are generated from a bivariate error-correction model with cointegration rank $r=1$ and lag order $k=2$:
\begin{equation}
\Delta X_t
=
\alpha(\cdot) \beta(\cdot)' X_{t-2}
+
\Gamma \Delta X_{t-1}
+
\mu
+
\varepsilon_t,
\qquad
t=3,\ldots,T,
\label{eq:sim_dgp_general}
\end{equation}
where $X_t=(X_{1t},X_{2t})'$, $\beta(\cdot)=(1,C(\cdot))'$, and
\[
\Pi(\cdot)=\alpha(\cdot)\beta(\cdot)'.
\]
The total sample size is set to
\[
T\in\{400,1000,2000,10000,20000\},
\]
with corresponding intervention date
\[
t^*\in\{200,500,1000,5000,10000\},
\]
So the post-intervention sample size is 
\[
N \in\{200,500,1000,5000,10000\}.
\]
The reported rejection rates are based on 1000 simulations and 2000 bootstrap samples at significance level $5\%$. For the null DGP, the error term $\varepsilon_t$ is Gaussian with covariance matrix
\begin{equation}
\Sigma_\varepsilon
=
\begin{pmatrix}
0.08 & 0.05\\
0.05 & 0.08
\end{pmatrix}.
\label{eq:sim_innovation_covariance}
\end{equation}
The common short-run dynamics and deterministic term are fixed across all scenarios:
\begin{equation}
\Gamma
=
\begin{pmatrix}
0.06 & 0.01\\
-0.01 & 0.05
\end{pmatrix},
\qquad
\mu
=
\begin{pmatrix}
0\\
0
\end{pmatrix}.
\label{eq:sim_short_run_parameters}
\end{equation}

The pre-intervention baseline parameters are
\begin{equation}
\alpha_0
=
\begin{pmatrix}
-0.12\\
0.08
\end{pmatrix},
\qquad
\beta_0
=
\begin{pmatrix}
1\\
-0.80
\end{pmatrix},
\qquad
\Pi_0=\alpha_0\beta_0'.
\label{eq:sim_baseline_parameters}
\end{equation}

\subsection{Scenario DGP design}
\label{subsec:simulation_scenarios}

Three scenarios are considered. The scenario 1 corresponds to exact post-intervention constancy, so that both the long-run and the short-run structure remain unchanged after the intervention. For $t\geq t^*$,
\begin{equation}
\alpha(\cdot)=\alpha_0=\begin{pmatrix}
-0.12\\
0.08
\end{pmatrix},
\qquad
C=C_0=-0.80,
\qquad
\beta(\cdot)=\beta_0=
\begin{pmatrix}
1\\
-0.80
\end{pmatrix}.
\label{eq:sim_scenario1}
\end{equation}
Hence, Scenario 1 satisfies the null of Test 1.\\

The scenario 2 is designed to represent smooth post-intervention structural evolution. For $t \geq t^*$, let
\begin{equation}
s_t
=0.2*
\left[
1-\cos\!\left(
\frac{t-t^*}{T-t^*}
\right)
\right]
\label{eq:sim_scenario2_transition}
\end{equation}
Then define
\begin{equation}
\alpha(\cdot)
=
\alpha_0
+ s_t
\left(
\alpha_1^{\text{smo}}-\alpha_0
\right),
\qquad
\alpha_1^{\text{smo}}
=
\begin{pmatrix}
-0.16\\
0.10
\end{pmatrix},
\label{eq:sim_scenario2_alpha}
\end{equation}
and
\begin{equation}
C(\cdot)
=
C_0
+
s_t
\left(
-1.20+1.00
\right),
\qquad
\beta(\cdot)=
\begin{pmatrix}
1\\
C(\cdot)
\end{pmatrix}.
\label{eq:sim_scenario2_beta}
\end{equation}
Thus, $\Pi(\cdot)=\alpha(\cdot)\beta(\cdot)'$ changes smoothly over the post-intervention period.\\

The scenario 3 is designed to represent a non-smooth structural break in the post-intervention long-run matrix. Let
\begin{equation}
t_{\text{break}}
=
t^*
+
0.3\,(T-t^*)
\label{eq:sim_scenario3_break}
\end{equation}
denote the time point of abrupt structural change over the post-intervention subsample. For $t^*<t\le t_{\text{break}}$, the system remains as the smooth structural change in Scenario 2:
\begin{equation}
\alpha(\cdot)=\alpha_0
+ s_t
\left(
\alpha_1^{\text{smo}}-\alpha_0
\right),
\qquad
C(\cdot)=-1.00
+
s_t
\left(
-1.20+1.00
\right).
\label{eq:sim_scenario3_left}
\end{equation}
For $t>t_{\text{break}}$, the post-break parameters shift abruptly to
\begin{equation}
\alpha(\cdot)
=
\alpha_1^{\text{Abr}}
=
\begin{pmatrix}
-0.24\\
0.20
\end{pmatrix},
\label{eq:sim_scenario3_alpha}
\end{equation}
and
\begin{equation}
C(\cdot)=-1.80,
\qquad
\beta(\cdot)=
\begin{pmatrix}
1\\
-1.80
\end{pmatrix}.
\label{eq:sim_scenario3_beta}
\end{equation}
Accordingly, $\Pi(\cdot)=\alpha(\cdot)\beta(\cdot)'$ exhibits an interior jump at $t_{\text{break}}$.\\

These three scenarios are chosen so that the two tests can be evaluated against distinct benchmark cases. Scenario 1 is the null for Test 1. Scenario 2 is a power design for Test 1 but a null design for Test 2, since it contains smooth post-intervention change without any non-smooth break. Scenario 3 is a power design for Test 2, since it contains an abrupt post-intervention level shift in the long-run matrix.\\

Hence, in the simulation results, Test 1 should display low rejection frequency in Scenario 1 and high rejection frequencies in Scenarios 2 amd 3. In contrast, Test 2 should display low rejection frequencies in Scenarios 1 and 2, but a high rejection frequency in Scenario 3.

\subsection{Selection of grid spacing, bandwidths and bootstrap block lengths}
\label{rem:simulation_tuning_parameters}
 
The parameters used in the simulation are chosen to follow the 
standard large-sample logic for kernel smoothing and block bootstrap methods, while remaining numerically stable in finite samples.\\

For the first-stage local Johansen estimation, the local window length $l^{\text{win}}_N$ is selected from a finite set of candidate values satisfying
\[
l^{\text{win}}_N\rightarrow\infty,
\qquad
\frac{l^{\text{win}}_N}{N}\rightarrow 0 .
\]
In the simulation, this is implemented by considering bandwidth candidates in a finite set. This rule follows the usual kernel-smoothing principle that the number of observations inside each local window should diverge, while the bandwidth fraction should vanish asymptotically. For local polynomial smoothing, a commonly used reference rate for the bandwidth fraction is of order $N^{-1/5}$, corresponding to a local window length of order $N^{4/5}$, although the finite-sample constants are typically selected by cross-validation or simulation calibration; see \shortciteA{fan1996local} and \shortciteA{fan1996study}.\\

In practice, the grid spacing $\Delta_N$ is selected so that each first-stage kernel window contains several grid intervals, typically corresponding to
\[
\Delta_N \approx \frac{l^{\text{win}}_N}{5}
\quad \text{to} \quad
\frac{l^{\text{win}}_N}{10}.
\]
This choice makes the grid sufficiently dense to estimate the post-intervention coefficient path, while keeping the number of local Johansen estimations reasonable.\\

For the smoothness test, the second-stage local-linear comparison is applied to the estimated path
\[
\widehat\zeta(g_m)
=
\operatorname{vec}\{\widehat\Pi(g_m)\}.
\]
Since this test concerns the regularity of the post-intervention path itself, 
the grid points are expressed on the post-intervention time scale
\[
u_m
=
\frac{t_m-t^*}{N}
\in [0,1],
\]
rather than on the full-sample scale $t_m/T$. Under this normalisation, the second-stage bandwidth $h_{\text{smo}}$ is directly interpretable as a fraction of the post-intervention period. In the simulations, $h_{\text{smo}}$ is selected from a moderate range on this post-intervention scale. Compactly supported kernels are preferred in the second stage, because the statistic is based on comparing pooled and split local-linear fits in a neighbourhood of each candidate point. This avoids excessively wide effective windows that may arise when a truncated Gaussian kernel is used.\\

The moving-block bootstrap block length is selected according to the standard dependent-data bootstrap requirement
\[
l^{\text{boots}}_N\rightarrow\infty,
\qquad
\frac{l^{\text{boots}}_N}{N}\rightarrow 0 .
\]
Rather than using a fixed fraction of $N$, the simulations use the rule
\[
l^{\text{boots}}_N
=
\max\left\{50,\left\lfloor c_l N^{1/3}\right\rfloor\right\},
\]
where $c_l$ is a moderate constant. This choice is consistent with the block-bootstrap literature, where the optimal block length depends on the target functional and may be of order $N^{1/3}$, $N^{1/4}$, or $N^{1/5}$; see \shortciteA{hall1995blocking}, \shortciteA{politis2004automatic}, and \shortciteA{lahiri2003resampling}. The leave-block length used in the first-stage bandwidth cross-validation is chosen on the same order as the bootstrap block length, typically proportional to $N^{1/3}$, so that the validation observation is separated from the observations used in local estimation without removing an excessively large part of the local window.\\

Finally, for the split local-linear fit in Test 2, admissible candidate points are required to have a minimum amount of kernel mass and a minimum number of grid points on both sides. This trimming rule prevents the split fit from being driven by boundary points or by few observations on either side of a 
candidate point.

\subsection{Results and discussion} \label{Results and discussion}

Now we discuss the results of simulation study on the finite-sample performance of two tests. Table~\ref{tab:Power} shows the empirical rejection rates of two tests under the null (size) and under post-intervention alternatives (power) with various sample sizes.\\

For both tests, the size reduces and power increases as the sample size increases. But there are some differences in their finite-sample performance. Test 1 reaches rejection frequencies close to the nominal 5\% level quickly at middle sample sizes and has higher power than test 2. Test 2 also shows decreasing size and increasing power, but its size remains above 0.05 in the reported samples. This is consistent with the fact that Test 2 is computed from the estimated path $\widehat{\Pi}(g_m)$, so its statistic includes both first-stage local estimation error and second-stage split-fit error.\\

The proposed tests differ from existing structural-change tests in cointegrated systems. \shortciteA{seo1998tests} considers tests for structural changes in the cointegrating vector and adjustment vector in an ECM with an unknown break point. These tests are designed for discrete structural changes. By contrast, this paper considers a known intervention time and allows post-intervention parameters to vary over rescaled time. The tests also differ from \shortciteA{phillips2017estimating}, where the null hypothesis is parametric, $H_0:f(\delta)=g(\delta;\theta_0)$, and the statistic is based on an $L^2$-type discrepancy between a kernel estimator and a parametric counterpart. In this paper, Test~1 examines whether the post-intervention parameters remain equal to the pre-intervention benchmark, while Test~2 examines whether the estimated post-intervention path is compatible with $C^1$-smooth evolution.\\

\begin{table}[]
\centering
\begin{tabular}{cccccccccc}
\multicolumn{10}{c}{Performance of Tests under different sample sizes}                           \\ \hline
              &  &       &     &  & \multicolumn{2}{c}{Test 1} &  & \multicolumn{2}{c}{Test 2} \\
\begin{tabular}[c]{@{}c@{}}Sample \\ size\end{tabular} &
   &
  \begin{tabular}[c]{@{}c@{}}Mean \\ Bandwidth \\ $h$\end{tabular} &
  \begin{tabular}[c]{@{}c@{}}Mean \\ Bandwidth \\ $h_{\text{smo}}$ \end{tabular} &
   &
  Size &
  Power &
   &
  Size &
  Power \\ \hline
N=200   &  & 0.400 & 0.600 &  & 0.101        & 0.252       &  & 0.084        & 0.127       \\
N=500  &  & 0.160 & 0.240 &  & 0.046        & 0.524       &  & 0.075        & 0.536       \\
N=1000 &  & 0.080 & 0.120 &  & 0.032        & 0.817       &  & 0.073        & 0.802       \\
N=5000 &  & 0.044 & 0.065 &  & 0.038        & 0.952       &  & 0.078       & 0.883       \\
N=10000 &  & 0.043 & 0.065 &  & 0.039        & 0.981       &  & 0.062        & 0.923       \\ \hline
\multicolumn{10}{c}{The significance level is 0.05.}                                          
\end{tabular}
\caption{Size and power of Tests with different sample sizes. The size of Test 1 is under DGP Scenario 1; the power of Test 1 and the size of Test 2 are under DGP Scenario 2; the power of Test 2 is under DGP Scenario 3.}
\label{tab:Power}
\end{table}

\newpage

\section{Empirical application} \label{Empirical3}

This section studies the change in the long-run cointegrating relationship among key energy, foreign-exchange, and gold-market indicators after the start of Russia’s invasion of Ukraine on 24 February 2022, which is treated as an exogenous intervention date $t^*$. The event provides a natural setting to evaluate whether a major geopolitical shock affects (i) the cointegrating relation among these markets and (ii) whether the change of the long-run relation is smooth. Understanding such changes is relevant for empirical asset-pricing work and for policy discussions, since energy-price shocks can affect inflation and external balances, exchange-rate dynamics can transmit global shocks across economies, and gold-market also reflects shifts in risk appetite.\\

The daily data used are from the Federal Reserve Bank of St.\ Louis FRED database over the sample horizon from 4 January 2016 to 1 January 2026. The vector of observed variables is
\begin{equation}
X_t=(X_{1t},X_{2t},X_{3t})'=\big(\ln P^{\text{Oil}}_
t,\ \ln P^{\text{US/EU}}_t,\ \ln P^{\text{AU}}_t\big)^\prime,
\end{equation}
where $P^{\text{Oil}}_
t$ is the spot price of Brent crude oil in U.S.\ dollars per barrel, $P^{\text{US/EU}}_t$ is the spot exchange rate (USD/EUR), and $P^{\text{AU}}_t$ is the Credit Suisse NASDAQ Gold FLOWS103 Price Index.\\

The choice of these three variables is closely related to the economic transmission channels activated by the outbreak of the Russia--Ukraine war on 24 February 2022. Brent crude oil is included to capture the energy-supply channel. Russia was one of the largest suppliers of oil and gas to world markets before the war, and the war raised concerns about sanctions, supply disruption, and the reallocation of energy trade. As a result, Brent prices provide a direct measure of how geopolitical conflict was priced into global energy markets. This is consistent with the oil-shock literature, which emphasises that oil-price movements may reflect precautionary demand and expectations about future supply disruptions \shortcite{kilian2009not}. The USD/EUR exchange rate is included to capture the currency and international adjustment channel. The war generated an asymmetric shock across advanced economies because Europe was more directly exposed to Russian energy supply than the United States. Changes in the dollar--euro exchange rate therefore reflect the relative reassessment of European and U.S. macroeconomic conditions, including energy-import costs, inflation expectations, monetary-policy responses, and global risk aversion. This is also consistent with evidence that exchange rates and commodity prices are closely linked \shortcite{chen2010exchange}, and with the literature on safe-haven currencies, which shows that some currencies may appreciate when global risk aversion rises \shortcite{habib2012safe}. The gold-related index is included to capture the safe-haven and store of-value channel. Gold is commonly regarded as a monetary asset and a store of value, especially during periods of political uncertainty, inflation concern, and financial-market stress. The outbreak of war increased uncertainty about global growth, energy costs, inflation, and financial stability. In this environment, gold-related prices can be interpreted as reflecting investors' demand for protection against adverse macro-financial conditions and heightened risk aversion. This interpretation is in line with studies showing that gold can act as a hedge or safe-haven asset in periods of market turmoil, although its safe-haven role is state-dependent and need not hold uniformly across all markets or all crises \shortcite{baur2010gold, baur2010safe}. \\

The three series are synchronised on a common daily calendar, with non-trading days and missing observations filtered to form daily observations. The methodology set out in the previous section is then applied to assess whether the cointegration parameters remain stable across the pre- and post-intervention periods, and to test any post-intervention smooth evolution in the cointegrating matrices. The kernel function chosen is the Gaussian kernel, and the bandwidth $h=0.075$, $h_{\text{smo}}=0.12$ are selected.\\

\begin{figure}[htbp]
    \centering
    \includegraphics[scale=0.46]{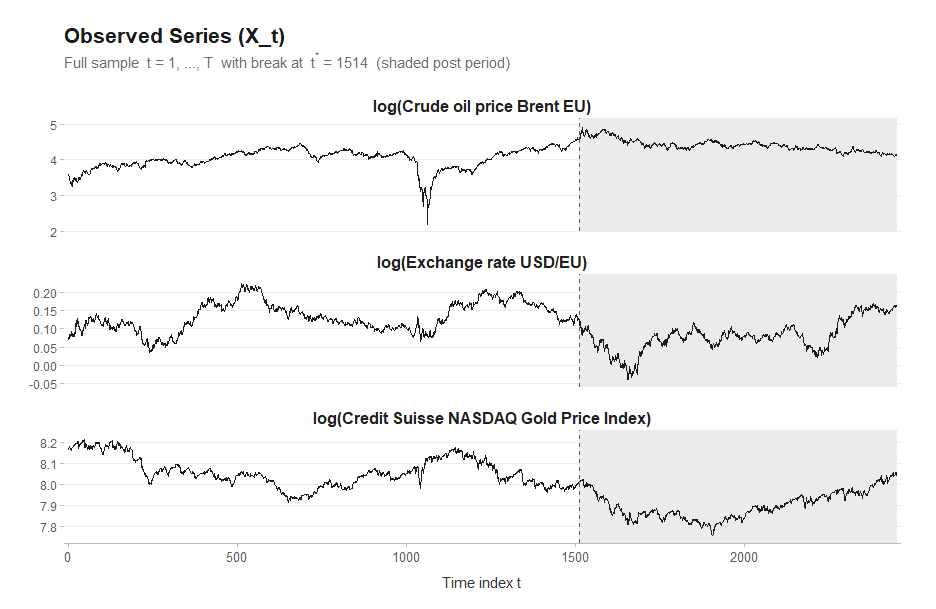}
    \caption{Series of $\ln P^{\text{Oil}}_
t$, $\ln P^{\text{US/EU}}_t$ and $\ln P^{\text{AU}}_t$ during 4 January 2016 to 1 January 2026.}
    \label{X_t}
\end{figure}

\newpage

\begin{table}[htbp]
\centering
\caption*{Pre-intervention unit-root diagnostics}
\small
\begin{tabular}{lccccccc}
\hline
& Level ADF & Level PP & Level KPSS & Diff. ADF & Diff. PP & Diff. KPSS &  \\
\hline
$X_1$ & -2.385 & -2.045 & 2.448 & -12.465 & -40.749 & 0.073 & $I(1)$ \\
$X_2$ & -2.211 & -2.357 & 3.032 & -14.485 & -37.075 & 0.104 & $I(1)$ \\
$X_3$ & -2.019 & -2.125 & 3.859 & -14.573 & -37.716 & 0.080 & $I(1)$ \\
\hline
\end{tabular}
\vspace{0.5em}
\caption{The pre-intervention sample is $t=1,\ldots,1514$. 
ADF and PP tests have the null hypothesis of a unit root, while the KPSS test has the null hypothesis of stationarity. 
At the 5\% level, the critical values are approximately $-2.86$ for ADF, $-2.864$ for PP, and $0.463$ for KPSS. 
The results indicate that the three original series are nonstationary in levels but stationary in their first difference.}
\label{tab:ch3_unit_root_diagnostics}
\end{table}

\begin{table}[htbp]
\centering
\caption*{Pre-intervention Johansen trace test}
\begin{tabular}{lccc}
\hline
Null hypothesis & Trace statistic & 5\% critical value &  \\
\hline
$r \leq 0$ & 35.666 & 31.520 & Reject \\
$r \leq 1$ & 14.051 & 17.950 & Accept  \\
$r \leq 2$ & 5.631 &  8.180 & Accept  \\
\hline
\end{tabular}
\vspace{0.5em}
\caption{The test is applied to the pre-intervention sample $t=1,\ldots,1514$ with lag order $k=2$. 
Critical values correspond to the unrestricted-constant specification. 
The formal trace test selects rank $r=1$. 
}
\label{tab:ch3_johansen_rank_test}
\end{table}

\begin{table}[htbp]
\centering
\caption*{Pre-intervention baseline estimates under $r=1$}
\begin{tabular}{lcc}
\hline
Variable & $\widehat\alpha_0$ & $\widehat\beta_0$ \\
\hline
$X_1$ & -0.0119 & 1.0000 \\
$X_2$ & -0.0004 & -2.4418 \\
$X_3$ & -0.0012 & 1.5337 \\
\hline
\end{tabular}
\vspace{0.5em}
\caption{The first-row normalised estimates are obtained from the pre-intervention observation under the tested rank $r=1$.}
\label{tab:ch3_preintervention_baseline_estimates}
\end{table}

\newpage

Figure \ref{X_t} demonstrates the trend of logarithm of Brent crude oil prices, the USD/EUR spot exchange rate, and the Gold price index over the full sample, with the intervention time point ($t^*$, Russian Invasion on 24 Feb. 2022). In the post-intervention period, the three series continue to co-move, but the amplitudes and persistence of fluctuations differ from the pre-intervention segment, suggesting that the adjustment dynamics of the system may have evolved.  \\

The preliminary diagnostics provide support for modelling the three-dimensional series $X_t=(X_{1t},X_{2t},X_{3t})'$ within an integrated time-series framework. As reported in Table~\ref{tab:ch3_unit_root_diagnostics}, the ADF and PP tests do not reject the unit-root null for the level series, whereas the KPSS test rejects the stationarity null in levels. After first differencing, the ADF and PP statistics reject the unit-root null and the KPSS statistics do not reject stationarity. These results indicate that all three components of $X_t$ are consistent with $I(1)$ behaviour in the pre-intervention sample, which justifies the use of a VECM-type long run specification.\\

The pre-intervention Johansen trace results in
Table~\ref{tab:ch3_johansen_rank_test} are used to guide the specification of the baseline long-run component. The formal trace statistic is large enough to reject the null of no cointegration at the 5\% level and select the pre-intervention rank $r=1$. The baseline estimates reported in Table~\ref{tab:ch3_preintervention_baseline_estimates} therefore provide the reference adjustment vector and cointegrating vector used in the subsequent post-intervention tests. Under this specification, the estimated loading vector
$\widehat\alpha_0$ is relatively small in magnitude, suggesting slow pre-intervention adjustment toward the estimated long-run relation. The normalised cointegrating vector $\widehat\beta_0$ defines the baseline long-run combination against which the post-intervention coefficient path is evaluated. Hence, the subsequent empirical tests should be read as conditional tests of post-intervention departures from this maintained pre-intervention long-run structure.\\

\begin{figure}[htbp]
    \centering
    \includegraphics[scale=0.4]{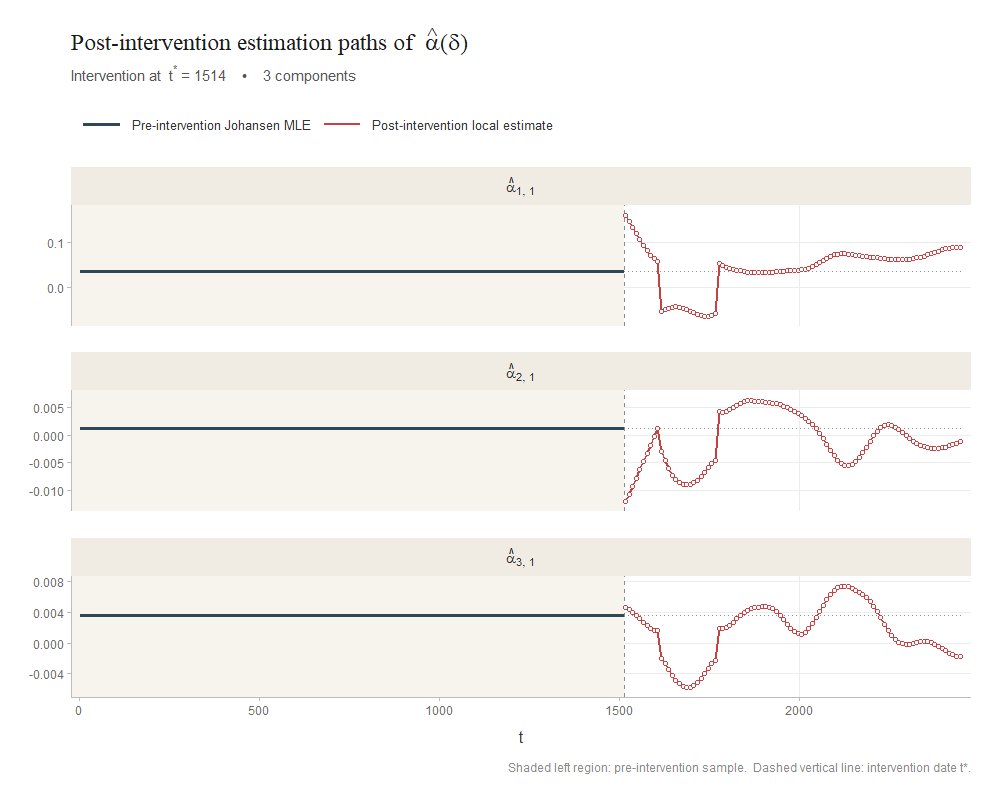}
    \caption{Pre-intervention adjustment coefficient vector $\alpha_0$ and post-intervention functional adjustment coefficient vector $\alpha(t/T)$ under unit-normalisation in $\beta(t/T)$ .}
    \label{alpha_t}
\end{figure}

\begin{figure}[htbp]
    \centering
    \includegraphics[scale=0.4]{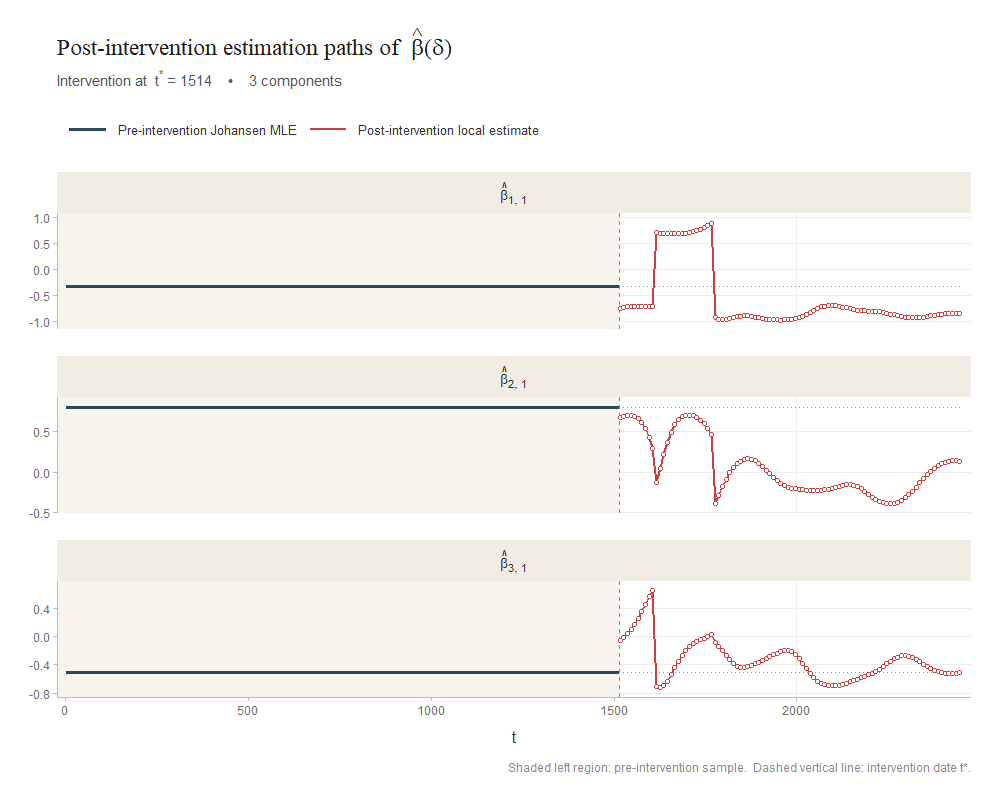}
    \caption{Pre-intervention cointegrating vector $\beta_0$ and post-intervention functional cointegrating vector $\beta(t/T)$ under unit-normalisation}
    \label{beta_t}
\end{figure}

\begin{landscape}
\begin{figure}[htbp]
    \centering
    \includegraphics[scale=0.5]{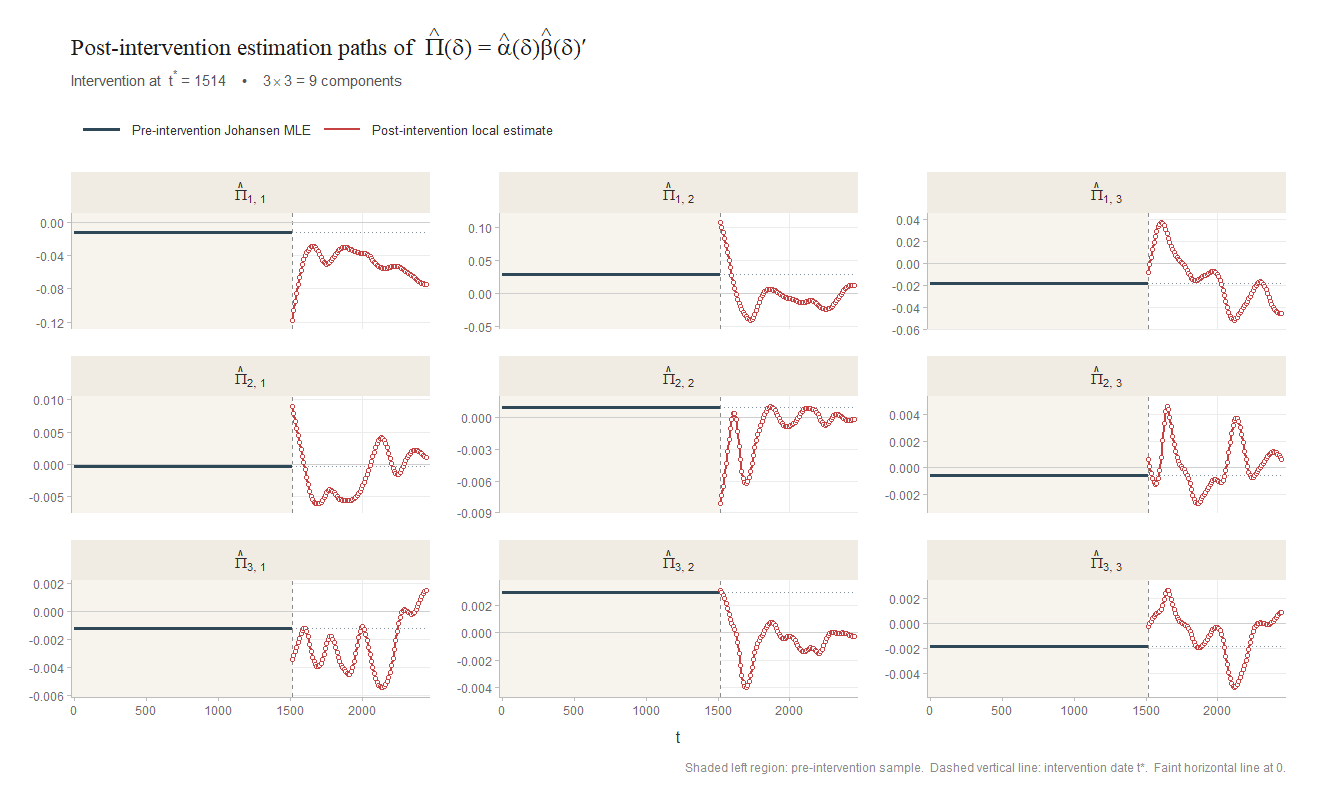}
    \caption{Pre-intervention cointegrating matrix $\Pi_0$ and post-intervention functional cointegrating matrix $\Pi(t/T)$}
    \label{Pi_t}
\end{figure}
\end{landscape}

\newpage

\begin{table}[]
\centering
\begin{tabular}{llll}
\multicolumn{4}{c}{Results of testing for functional cointegrating parameters} \\ \hline
                                     &              &             &           \\ \hline
\begin{tabular}[c]{@{}l@{}}$T = 2453$, $t^* = 1514$,\\ $N=939$, $\delta^* = 0.6172$, \\$p=3$, $k=3$, $r=1$, \\ Bandwidth $h = 0.075$, \\ $h_{\text{smo}}=0.12$,\\ boostrap block length = 24 \end{tabular} &
   &
   &
   \\ \hline
                                     &              &             &           \\
 &
  \begin{tabular}[c]{@{}l@{}}Statistics\end{tabular} &
  Critical value &
  $p$-value \\
                                     &              &             &           \\
$H_0^{\text{Base}}$: $\alpha(\delta)=\alpha_0$,\\        \ \ \ \ \ \ \ \ \ $\beta(\delta)=\beta_0$     & 0.00139     & 0.00133     & 0.036 **     \\
                                     &              &             &           \\
\begin{tabular}[c]{@{}l@{}}$H_0^{\text{smo}}$: $\Pi(\delta)$ is $C^1$ smooth \end{tabular} &
  0.01413 &
  0.05696 &
  0.591 \\
                                     &              &             &           \\ \hline
\multicolumn{4}{c}{The significance level is 0.05}                           
\end{tabular}
\caption{Testing results of  functional cointegrating parameters.}
\label{2test_result}
\end{table}

Figures \ref{alpha_t} and \ref{beta_t} compare the constant-parameter benchmark estimates $(\alpha_0,\beta_0)$ with the post-intervention functional estimates $(\alpha(\delta),\beta(\delta))$. The time-varying adjustment coefficients indicate that the speed and direction of error correction changes through the post-intervention window. This means that the channels through which deviations from the long-run equilibrium are absorbed are not stable after $t^*$: some variables exhibit a larger share of the re-equilibration at particular times, while others become comparatively less responsive. This pattern is consistent with an environment in which risk appetite and hedging demand shift gradually after a major geopolitical shock. Figure \ref{beta_t} shows that the estimated cointegrating vector also exhibits smooth time variation in the post-intervention period. While the normalisation fixes the first element of $\beta(\delta)$, the remaining smoothly changing over time, implying that the relative contribution of the exchange rate and gold-market uncertainty to the long-run equilibrium relation is not constant after $t^*$.\\

Table \ref{2test_result} reports the two tests for functional cointegrating parameters. For Test 1, the $p$-value provide evidence against equality with the pre-intervention benchmark at conventional levels. It suggests that the deviations of the post-intervention estimates are sufficiently persistent and large, relative to sampling uncertainty, to conclude that the system moved to a different parameter configuration. For Test 2, the $p$-value does not provide evidence against the null of a smooth post-intervention parameter path. Taken together, the two tests suggest that the post-intervention cointegrating structure differs from the pre-intervention benchmark, and there is no statistical evidence of interior jumps or kinks in structural changes in the post-intervention path.\\

The rejection of the pre-intervention benchmark suggests that the war altered the long-run configuration linking energy prices, exchange rate, and gold-related
safe-haven demand. This implies that the transmission of geopolitical shocks cannot be assessed only through contemporaneous movements in oil prices or exchange rates. A persistent change in the cointegrating structure indicates that the relative roles of energy costs, currency adjustment, and portfolio hedging may shift after the shock. At the same time, accepting
$C^1$-smoothness suggests that the adjustment was not characterised by a sequence of sharp internal breaks, but by gradual re-pricing over the post-intervention period. This is relevant for monetary and financial-stability monitoring, since energy-driven inflation pressure, exchange-rate movements and safe-haven demand may be transmitted through a slowly changing system rather than through a one-off discrete regime shift. The result is consistent with the policy concern that the Russia--Ukraine war generated a persistent energy-market shock and heightened volatility in related financial markets
\shortcite{iea2022weo,ecb2022energy}.\\

Our results also align with the long-standing view that oil prices and exchange rates can share a stable long-run association, even though short-run transmission mechanisms may change across episodes \shortcite{chen2007oil}. At the same time, the finding of post-intervention effects in adjustment dynamics is consistent with recent evidence that studies volatility and time-varying spillovers around the Russia-Ukraine war. For example, \shortciteA{izzeldin2023impact} document volatility responses in European and global markets and commodity groups around the invasion. Furthermore, the behaviour of the gold-volatility component is also consistent with the broader hedging and safe-haven theory. Gold exhibits safe-haven properties in several developed equity markets during crisis periods \shortcite{baur2010gold}, and can act as both a hedge and a safe haven against drastic USD exchange-rate movements \shortcite{reboredo2013gold}. Time variation in gold ETF volatility and adjustment coefficient can be read as shifts in the market price of uncertainty and hedging demand after the invasion. More broadly, evidence using wavelet coherence methods suggests that currencies and precious metals respond to geopolitical risk at relatively short horizons \shortcite{bkedowska2022hedging}, which is consistent with the evidence of post-intervention time-varying structural change here.

\section{Conclusion} \label{Conclusion3}

This paper proposes an approach to modelling and testing smooth structural change in cointegrated systems following a known intervention time. The main idea is to use the pre-intervention subsample to identify a stable VECM structure, while allowing the long-run parameters to evolve continuously after the intervention. By decomposing the cointegrating matrix into adjustment and cointegrating components and estimating them locally by kernel-weighted Johansen procedures, the framework delivers functional estimates of the long-run structural evolution.\\

I also introduce two tests that address empirical questions. Test~1 evaluates whether post-intervention parameters equal the pre-intervention benchmark, which is directly relevant for assessing whether an intervention induces a sustained change in the long-run cointegration relation. Test~2 evaluates whether the post-intervention cointegrating matrices are smoothly changing. Together, the two tests separate baseline deviation from post-intervention evolution, which is difficult to recover from standard discrete-break specifications when changes are gradual. The Monte Carlo results suggest that the tests become more reliable as the post-intervention sample size increases. Test 1 is Test 1 approaches the nominal 5\% size for large samples, and attains high power more quickly as the sample grows. Test 2 also becomes more informative with larger samples, but its size distortion remains somewhat oversized in the reported designs. This difference is consistent with the general difficulty of inference for smooth post-intervention evolution. \\

The empirical example studies the structural change in cointegrated system of daily log Brent crude oil prices, log USD/EUR spot exchange rate, and log Credit Suisse NASDAQ Gold FLOWS103 Price Index over 2016--2026, with the intervention of Russia's attack against Ukraine on 24 February 2022. Test 1 rejects that post-intervention parameters and the pre-intervention benchmark are at the same level, suggesting sufficient evidence of a significant different long-run parameter configuration. In contrast, Test 2 does not reject post-intervention $C^1$ smoothness of structural change, indicating that the post-intervention dynamics are better characterised by smooth evolution rather than discrete changes. The estimated functional coefficient path plots are consistent with this interpretation, showing relatively smooth evolution in cointegrating matrices over the post-intervention period.\\

From a policy perspective, the empirical results suggest that geopolitical events may change the long-run relationship among markets. In the Russia--Ukraine case, the post-intervention departure from the pre-intervention benchmark points to a reconfiguration of the relationship between energy prices, the dollar--euro exchange rate, and gold-related safe-haven demand. For central banks and financial authorities, this highlights the importance of monitoring cross-market long-run relations when assessing imported inflation, external-price pressures, and financial-market stress. The smooth nature of the estimated adjustment further suggests that policy responses should not rely only on detecting abrupt breaks immediately after the event. Instead, persistent
changes in the long-run adjustment path may reveal how markets gradually absorb geopolitical risk, sanctions, energy-supply uncertainty and portfolio reallocation. The proposed framework can therefore be used as an empirical monitoring device for distinguishing temporary market turbulence from more durable changes in cross-market transmission.\\

Overall, the approach provides an interpretable tool for analysing cointegrated systems subject to gradual post-intervention structural change. Several extensions may further strengthen the approach. First, data-driven bandwidth selection for  fit-loss smoothness test statistic would improve empirical implementation in finite samples. Moreover, the framework can be extended to allow for multiple interventions, unknown intervention time point, or variation in the cointegration rank. 

\newpage

\bibliography{bib}
\bibliographystyle{apacite}

\newpage

\section*{Appendix A}
\label{AppendixC}

\subsection*{Proof of Lemma
\ref{lem:local-population-equivalence}}

\begin{proof}
The lemma has three parts.

\medskip
\noindent
\textit{Part (i): local approximation of $\Pi(\delta)$.}

Since $K(\cdot)$ is bounded and either compactly supported or has sufficiently thin tails, observations with non-negligible local weight satisfy
\[
\left|\frac{t}{T}-\delta_0\right|=\left|\delta-\delta_0\right|=O(h).
\]
For compactly supported kernels this follows directly from the support condition. For kernels with unbounded support, the same order is understood in the effective-weight sense after neglecting tail contributions.\\

By Assumption 3.1,
\[
\|\alpha(\delta)-\alpha(\delta_0)\|=O(h^\gamma),
\qquad
\|C(\delta)-C(\delta_0)\|=O(h^\gamma)
\]
uniformly over the local window. Under the normalisation
\[
\beta(\delta)=
\begin{pmatrix}
I_r\\
C(\delta)
\end{pmatrix},
\]
this also implies
\[
\|\beta(\delta)-\beta(\delta_0)\|=O(h^\gamma).
\]
Using
\[
\Pi(\delta)=\alpha(\delta)\beta(\delta)',
\]
write
\begin{align}
\Pi(\delta)-\Pi(\delta_0)
&=
\alpha(\delta)\beta(\delta)'
-
\alpha(\delta_0)\beta(\delta_0)' \notag\\
&=
\alpha(\delta)
\{\beta(\delta)-\beta(\delta_0)\}'
+
\{\alpha(\delta)-\alpha(\delta_0)\}
\beta(\delta_0)' .
\label{eq:appC-pi-decomposition}
\end{align}
The functions $\alpha(\cdot)$ and $\beta(\cdot)$ are bounded on $[\delta^\ast,1]$. Therefore both terms on the right-hand side of \eqref{eq:appC-pi-decomposition} are $O(h^\gamma)$, and hence
\[
\|\Pi(\delta)-\Pi(\delta_0)\|
=
O(h^\gamma).
\]
This proves \eqref{eq:local-pi-approximation}.\\

\medskip
\noindent
\textit{Part (ii): consistency of the local residualisation coefficients.}

Consider $\widehat B_0(\delta_0)$. By definition,
\[
\widehat B_0(\delta_0)
=
\widehat P_{ZZ}(\delta_0)^{-1}
\widehat P_{Z0}(\delta_0),
\qquad
B_0^\circ(\delta_0)
=
P^\circ_{ZZ}(\delta_0)^{-1}
P^\circ_{Z0}(\delta_0).
\]
The local matrices are weighted sums over a window containing $O(Nh)$ effective observations. Since $Nh\to\infty$, $h\to 0$, $\gamma\in(1/2,1]$,  the weighted law of large numbers gives
\[
\widehat P_{ZZ}(\delta_0)-P^\circ_{ZZ}(\delta_0)
=
O_p((Nh)^{-1/2})+O(h^\gamma)
=
o_p(1),
\]
and similarly
\[
\widehat P_{Z0}(\delta_0)-P^\circ_{Z0}(\delta_0)
=
O_p((Nh)^{-1/2})+O(h^\gamma)
=
o_p(1).
\]
The term $O_p((Nh)^{-1/2})$ is the sampling fluctuation from $O(Nh)$ local observations, while $O(h^\gamma)$ is the local smoothing bias induced by replacing $\Pi(t/T)$ with $\Pi(\delta_0)$ in the local window.\\

Under the local nonsingularity condition for $P_{ZZ}(\delta_0)$,
\[
\widehat P_{ZZ}(\delta_0)^{-1}
-
P^\circ_{ZZ}(\delta_0)^{-1}
=
o_p(1).
\]
Therefore
\[
\widehat B_0(\delta_0)-B_0^\circ(\delta_0)
=
o_p(1).
\]
Similarly, replacing $\widehat P_{Z0}$ and $P_{Z0}$ by $\widehat P_{Z1}$ and $P_{Z1}$, gives
\[
\widehat B_1(\delta_0)-B_1^\circ(\delta_0)
=
o_p(1).
\]
This proves \eqref{eq:local-regression-consistency}.\\

\medskip
\noindent
\textit{Part (iii): asymptotic equivalence of feasible and population residuals.}

From the definitions of the feasible and population residuals,
\[
R_{0t}(\delta_0)-R_{0t}^\circ(\delta_0)
=
-
\{\widehat B_0(\delta_0)-B_0^\circ(\delta_0)\}'Z_t .
\]
Hence
\[
\|R_{0t}(\delta_0)-R_{0t}^\circ(\delta_0)\|
\le
\|\widehat B_0(\delta_0)-B_0^\circ(\delta_0)\|
\|Z_t\|.
\]
The vector $Z_t$ contains lagged differences $\Delta X_{t-i}$ and a deterministic constant. Under the post-intervention short-run stability condition and the moment conditions imposed on the innovation process, these lagged differences are stationary or asymptotically tight. Therefore,
\[
\max_{t:\,K_{t,h}(\delta_0)\ne0}\|Z_t\|=O_p(1).
\]
Combining this bound with
\[
\widehat B_0(\delta_0)-B_0^\circ(\delta_0)=o_p(1)
\]
yields
\[
\max_{t:\,K_{t,h}(\delta_0)\ne0}
\|R_{0t}(\delta_0)-R_{0t}^\circ(\delta_0)\|
=
o_p(1).
\]
The proof for $R_{1t}(\delta_0)-R_{1t}^\circ(\delta_0)$ is similar, using
\[
R_{1t}(\delta_0)-R_{1t}^\circ(\delta_0)
=
-
\{\widehat B_1(\delta_0)-B_1^\circ(\delta_0)\}'Z_t .
\]
This proves \eqref{eq:local-y-residual-equivalence} and \eqref{eq:local-x-residual-equivalence}.
\end{proof}

\subsection*{Proof of Theorem \ref{thm:consistency}}

\begin{proof}
Fix an interior post-intervention point $\delta_0\in(\delta^\ast,1)$. 
For notational simplicity, all local covariance matrices here are evaluated at $\delta_0$:
\[
S_{ij}:=S_{ij}(\delta_0), \qquad S^{\circ}_{ij}:=S^{\circ}_{ij}(\delta_0), \qquad i,j\in\{0,1\}.
\]

By Lemma \ref{lem:local-population-equivalence}, the feasible residuals and the corresponding population residuals are asymptotically equivalent in the local window. Hence, for each $i,j\in\{0,1\}$,
\[
\widehat S_{ij}-S_{ij}^{\circ}
=
O_p((Nh)^{-1/2})+O(h^\gamma)
=
o_p(1).
\tag{C.1}
\label{eq:appC-Sij-consistency}
\]
The term $O_p((Nh)^{-1/2})$ is the sampling error from $O(Nh)$ effective local observations, while $O(h^\gamma)$ is the local smoothing bias from replacing the time-varying coefficient matrix $\Pi(t/T)$ by $\Pi(\delta_0)$ inside the bandwidth-$h$ neighbourhood. Since $Nh\to\infty$ and $h\to0$, both terms vanish.\\

The sample and population matrix pair for the local generalised eigenvalue problem are
\[
(\widehat S_{10}\widehat S_{00}^{-1}\widehat S_{01}
,
\widehat S_{11}),
\]
and
\[
(S_{10}^{\circ}(S_{00}^{\circ})^{-1}S_{01}^{\circ}
,
S_{11}^{\circ}).
\]
Because $S_{00}^{\circ}$ is nonsingular and $\widehat S_{00}-S_{00}^{\circ}=o_p(1)$, the inverse map is continuous at $S_{00}^{\circ}$, and therefore
\[
\widehat S_{00}^{-1}-(S_{00}^{\circ})^{-1}=o_p(1).
\tag{C.2}
\label{eq:appC-S00-inverse}
\]
Using \eqref{eq:appC-Sij-consistency} and \eqref{eq:appC-S00-inverse},
\[
\widehat S_{11}-S_{11}^{\circ}=o_p(1).
\tag{C.3}
\label{eq:appC-matrix-pair-convergence}
\]
Write
\begin{align}
\widehat S_{10}\widehat S_{00}^{-1}\widehat S_{01}-S_{10}^{\circ}(S_{00}^{\circ})^{-1}S_{01}^{\circ}
&=
(\widehat S_{10}-S_{10}^{\circ})\widehat S_{00}^{-1}\widehat S_{01}
+
S_{10}^{\circ}\{\widehat S_{00}^{-1}-(S_{00}^{\circ})^{-1}\}\widehat S_{01}
\notag\\
&\quad+
S_{10}^{\circ}(S_{00}^{\circ})^{-1}(\widehat S_{01}-S_{01}^{\circ}),
\tag{C.4}
\label{eq:appC-A-decomposition}
\end{align}
where each term on the right-hand side of \eqref{eq:appC-A-decomposition} is $o_p(1)$. Hence \[
\widehat S_{10}\widehat S_{00}^{-1}\widehat S_{01}-S_{10}^{\circ}(S_{00}^{\circ})^{-1}S_{01}^{\circ}=o_p(1).
\]

Since $S_{11}^{\circ}$ is positive definite, and $\widehat S_{11} \xrightarrow{p} S_{11}^{\circ}$, $\widehat S_{11}$ is positive definite with probability approaching one.  Therefore the generalised eigenvalue problem can be transformed into an ordinary symmetric eigenvalue problem. Define
\[
\widehat{\Lambda}
:=
(\widehat S_{11})^{-1/2}
\widehat S_{10}(\widehat S_{00})^{-1}\widehat S_{01}
(\widehat S_{11})^{-1/2},
\qquad
\Lambda^{\circ}
:=
(S_{11}^{\circ})^{-1/2}
S_{10}^{\circ}(S_{00}^{\circ})^{-1}S_{01}^{\circ}
(S_{11}^{\circ})^{-1/2}.
\]
By \eqref{eq:appC-matrix-pair-convergence} and continuity of the matrix square-root inverse at the positive definite matrix $S_{11}^{\circ}$,
\[
\widehat{\Lambda}-\Lambda^{\circ}=o_p(1).
\tag{C.5}
\label{eq:appC-Lambda-convergence}
\]

Let $\Lambda^{\circ,\dagger}$ denote the spectral projection of $\Lambda^{\circ}$ associated with its first $r$ eigenvalues, and let $\widehat{\Lambda}^\dagger$ denote the corresponding sample spectral projection of $\widehat{\Lambda}$. By assumption, the $r$th and $(r+1)$th population generalised eigenvalues are separated:
\[
\lambda_r^{\circ}(\delta_0)-\lambda_{r+1}^{\circ}(\delta_0)\ge c_\lambda>0.
\tag{C.6}
\label{eq:appC-eigen-gap}
\]
This condition implies that the population $r$-dimensional eigenspace is isolated. Applying the standard eigenspace perturbation bound to $\widehat{M}$ and $M^{\circ}$,
\[
\|\widehat{\Lambda}^\dagger-\Lambda^{\circ,\dagger}\|
\le
\mathcal C
\frac{\|\widehat{\Lambda}-\Lambda^{\circ}\|}{c_\lambda}
=
o_p(1),
\tag{C.7}
\label{eq:appC-projection-convergence}
\]
where $\mathcal C<\infty$ is a constant independent of $N$. Hence the estimated $r$-dimensional eigenspace converges in probability to the population $r$-dimensional eigenspace.\\

The eigenspace in \eqref{eq:appC-projection-convergence} is the local cointegrating space. Therefore,
\[
\operatorname{span}\{\widetilde\beta(\delta_0)\}
\xrightarrow{p}
\operatorname{span}\{\beta(\delta_0)\},
\tag{C.8}
\label{eq:appC-space-consistency}
\]
where $\widetilde\beta(\delta_0)$ denotes any sample eigenvector matrix before imposing the identifying normalisation. The convergence in \eqref{eq:appC-space-consistency} is convergence of spaces, so it still allows postmultiplication by a nonsingular $r\times r$ matrix.\\

To obtain convergence of the normalised cointegrating matrix, write the population normalisation as
\[
\beta(\delta_0)
=
\begin{pmatrix}
I_r\\
C(\delta_0)
\end{pmatrix}.
\]
Let $\widetilde\beta_1(\delta_0)$ denote the upper $r\times r$ block of $\widetilde\beta(\delta_0)$. Since the upper block of the population matrix is $I_r$, the matrix $\widetilde\beta_1(\delta_0)$ is nonsingular with probability approaching one. Define the normalised estimator by
\[
\widehat\beta(\delta_0)
:=
\widetilde\beta(\delta_0)\widetilde\beta_1(\delta_0)^{-1}.
\]
The mapping from a basis of the local cointegrating space to the normalised matrix $(I_r,\widehat C(\delta_0)')'$ is continuous in a neighbourhood of the true space. Hence \eqref{eq:appC-space-consistency} implies
\[
\widehat\beta(\delta_0)\xrightarrow{p}\beta(\delta_0).
\tag{C.9}
\label{eq:appC-beta-consistency}
\]
Equivalently,
\[
\widehat C(\delta_0)\xrightarrow{p}C(\delta_0).
\tag{C.10}
\label{eq:appC-C-consistency}
\]

It remains to prove consistency of $\widehat\alpha(\delta_0)$. The estimator is
\[
\widehat\alpha(\delta_0)
=
\widehat S_{01}(\delta_0)\widehat\beta(\delta_0)
\left[
\widehat\beta(\delta_0)'
\widehat S_{11}(\delta_0)
\widehat\beta(\delta_0)
\right]^{-1}.
\]
Using \eqref{eq:appC-Sij-consistency} and \eqref{eq:appC-beta-consistency},
\[
\widehat S_{01}(\delta_0)\widehat\beta(\delta_0)
\xrightarrow{p}
S_{01}^{\circ}(\delta_0)\beta(\delta_0),
\tag{C.11}
\label{eq:appC-alpha-numerator}
\]
and
\[
\widehat\beta(\delta_0)'
\widehat S_{11}(\delta_0)
\widehat\beta(\delta_0)
\xrightarrow{p}
\beta(\delta_0)'S_{11}^{\circ}(\delta_0)\beta(\delta_0).
\tag{C.12}
\label{eq:appC-alpha-denominator}
\]
By the local identification condition,
\[
\beta(\delta_0)'S_{11}^{\circ}(\delta_0)\beta(\delta_0)
\]
is nonsingular. Hence, by continuity of the inverse map,
\[
\left[
\widehat\beta(\delta_0)'
\widehat S_{11}(\delta_0)
\widehat\beta(\delta_0)
\right]^{-1}
\xrightarrow{p}
\left[
\beta(\delta_0)'S_{11}^{\circ}(\delta_0)\beta(\delta_0)
\right]^{-1}.
\tag{C.13}
\label{eq:appC-alpha-inverse}
\]
Combining \eqref{eq:appC-alpha-numerator}--\eqref{eq:appC-alpha-inverse},
\[
\widehat\alpha(\delta_0)
\xrightarrow{p}
S_{01}^{\circ}(\delta_0)\beta(\delta_0)
\left[
\beta(\delta_0)'S_{11}^{\circ}(\delta_0)\beta(\delta_0)
\right]^{-1}
=
\alpha(\delta_0).
\tag{C.14}
\label{eq:appC-alpha-consistency}
\]
Therefore,
\[
\widehat\beta(\delta_0)\xrightarrow{p}\beta(\delta_0),
\qquad
\widehat\alpha(\delta_0)\xrightarrow{p}\alpha(\delta_0).
\]
\end{proof}

\subsection*{Proof of Theorem \ref{thm:beta-convergence}: asymptotic convergence of $\widehat\beta(\delta_0)$}

\begin{proof}
Write, for notational simplicity,
\[
\vartheta_0:=\vartheta(\delta_0),
\qquad
\mathcal D_N:=\mathcal D_{N,\delta_0}^{\vartheta},
\qquad
\mathcal Q:=\mathcal Q_{\delta_0}^{\vartheta}.
\]
Define the centred local criterion process in rotated and scaled coordinates:
\[
\varrho_N(z)
=
\widehat{\mathcal L}_{\delta_0}
\left(
\vartheta_0+
\mathcal Q\mathcal D_N^{-1}z
\right)
-
\widehat{\mathcal L}_{\delta_0}(\vartheta_0),
\qquad
z\in\mathbb R^m.
\]
By Assumption 3.6,
\begin{align}
\varrho_N(z)
&=
-
z'\mathcal D_N^{-1}
\mathcal Q'
\mathcal S_N^\vartheta(\delta_0)
-
\frac12
z'\mathcal D_N^{-1}
\mathcal Q'
\mathcal H_N^\vartheta(\delta_0)
\mathcal Q
\mathcal D_N^{-1}z
+
o_p(1),
\label{eq:appC-local-criterion-expansion}
\end{align}
uniformly on every compact subset of $\mathbb R^m$.\\

By the high-level score and Hessian limits,
\[
\mathcal D_N^{-1}
\mathcal Q'
\mathcal S_N^\vartheta(\delta_0)
=
O_p(1),
\]
and
\[
\mathcal D_N^{-1}
\mathcal Q'
\mathcal H_N^\vartheta(\delta_0)
\mathcal Q
\mathcal D_N^{-1}
=
O_p(1).
\]
Moreover, the limiting Hessian matrix is nonsingular with probability approaching one. Hence the local criterion in \eqref{eq:appC-local-criterion-expansion} is asymptotically quadratic with tight linear and curvature terms.\\

Let
\[
\widehat z_N
=
\mathcal D_N
\mathcal Q'
\{\widehat\vartheta(\delta_0)-\vartheta(\delta_0)\}.
\]
Since $\widehat\vartheta(\delta_0)$ minimises the local criterion, $\widehat z_N$ minimises the criterion in the local coordinates up to an asymptotically negligible error. The quadratic expansion above implies that $\widehat z_N$ is tight. Otherwise, along a sequence with $\|\widehat z_N\|\to\infty$, the positive quadratic term would dominate the linear term on compactly expanding sets, contradicting local minimality. Therefore,
\[
\mathcal D_{N,\delta_0}^{\vartheta}
\mathcal Q_{\delta_0}^{\vartheta\prime}
\{
\widehat\vartheta(\delta_0)-\vartheta(\delta_0)
\}
=
O_p(1).
\]
Hence the block decomposition
\[
\mathcal Q_{\delta_0}^{\vartheta}
=
\left(
\mathcal Q_{\delta_0}^{\vartheta,f},
\mathcal Q_{\delta_0}^{\vartheta,s}
\right),
\qquad
\mathcal D_{N,\delta_0}^{\vartheta}
=
\operatorname{diag}
\{N\sqrt h\,I_r,\ Nh\,I_{m-r}\}
\]
give
\[
\mathcal Q_{\delta_0}^{\vartheta,f\prime}
\{
\widehat\vartheta(\delta_0)-\vartheta(\delta_0)
\}
=
O_p((N\sqrt h)^{-1}),
\]
and
\[
\mathcal Q_{\delta_0}^{\vartheta,s\prime}
\{
\widehat\vartheta(\delta_0)-\vartheta(\delta_0)
\}
=
O_p((Nh)^{-1}).
\]
Since the map $C\mapsto\beta=(I_r,C')'$ is linear under the imposed normalisation, the same directional rates apply to the corresponding free block of $\widehat\beta(\delta_0)-\beta(\delta_0)$.
\end{proof}

\subsection*{Order calculations used in Theorem \ref{thm:beta-convergence}}

The local window contains $O(Nh)$ effective observations. Under the functional central limit theorem,
\[
N^{-1/2}X_{\lfloor Ts\rfloor}
\Rightarrow
W(s), \qquad s \in [t^*/T,1]
\]
so the process variable $X_t$, and the population residuals $R_{1t}^\circ(\delta_0)$, is of order $O_p(\sqrt N)$ in the local window.\\

Therefore each outer product
\[
R_{1t}^\circ(\delta_0)R_{1t}^\circ(\delta_0)'
\]
is of order $O_p(N)$. Summing over $O(Nh)$ local observations gives
\[
G_N^x(\delta_0)
=
\sum_{t=t^*+1}^{T}
K_{t,h}(\delta_0)
R_{1t}^\circ(\delta_0)R_{1t}^\circ(\delta_0)'
=
O_p(N^2h).
\]
The parameter-space Hessian is induced by this level signal through the normalisation map:
\[
\mathcal H_N^\vartheta(\delta_0)
=
J_{\delta_0}'
\{I_r\otimes G_N^x(\delta_0)\}
J_{\delta_0}
+
o_p(N^2h).
\]
Hence the leading Hessian is of order $O_p(N^2h)$ in the dominant direction.\\

The score has the form of a local sum of products between the residualised level variable and a stationary innovation term:
\[
\mathcal S_N^\vartheta(\delta_0)
\approx
\sum_{t=t^*+1}^{T}
K_{t,h}(\delta_0)
R_{1t}^\circ(\delta_0)\varepsilon_t .
\]
Each summand has variance of order $N$, since $R_{1t}^\circ(\delta_0)=O_p(\sqrt N)$ and $\varepsilon_t=O_p(1)$. With $O(Nh)$ weakly dependent mean-zero summands, the standard deviation of the sum is
\[
\{Nh\cdot N\}^{1/2}
=
N\sqrt h.
\]
Thus
\[
\mathcal S_N^\vartheta(\delta_0)
=
O_p(N\sqrt h).
\]

These orders explain the fast rate
\[
\frac{N\sqrt h}{N^2h}
=
\frac{1}{N\sqrt h}.
\]
However, because the local integrated regressor is dominated by the Brownian level at $\delta_0$, the unrotated local signal matrix is asymptotically close to rank one. After rotation, the direction associated with the local Brownian level has scaling $N\sqrt h$, while the orthogonal directions are driven by local Brownian increments and have scaling $Nh$. This motivates
\[
\mathcal D_{N,\delta_0}^{\vartheta}
=
\operatorname{diag}
\{N\sqrt h\,I_r,\ Nh\,I_{m-r}\}.
\]

\subsection*{Proof of Theorem \ref{thm:alpha-convergence}: asymptotic convergence of $\widehat\alpha(\delta_0)$}

\begin{proof}
Recall the map
\[
\Psi(S_{01},S_{11},\beta)
=
S_{01}\beta(\beta'S_{11}\beta)^{-1}.
\]
Then
\[
\widehat\alpha(\delta_0)
=
\Psi\left(
\widehat S_{01}(\delta_0),
\widehat S_{11}(\delta_0),
\widehat\beta(\delta_0)
\right).
\]
Take a first-order expansion of $\Psi$ around
\[
\left(
S_{01}(\delta_0),
S_{11}(\delta_0),
\beta(\delta_0)
\right).
\]
This gives
\begin{align}
\widehat\alpha(\delta_0)-\alpha(\delta_0)
&=
D_{S_{01}}\Psi
\left[
\widehat S_{01}(\delta_0)-S_{01}(\delta_0)
\right]
\notag\\
&\quad+
D_{S_{11}}\Psi
\left[
\widehat S_{11}(\delta_0)-S_{11}(\delta_0)
\right]
\notag\\
&\quad+
D_{\beta}\Psi
\left[
\widehat\beta(\delta_0)-\beta(\delta_0)
\right]
+
\rho_N,
\label{eq:appC-alpha-expansion}
\end{align}
where $\rho_N$ is the second-order remainder.\\

The first two terms in \eqref{eq:appC-alpha-expansion} are driven by stationary local moments. Once the cointegrating vector is imposed, the error-correction term is stationary, and the relevant local moments are averages over $O(Nh)$ effective observations. Therefore their stochastic order is
\[
O_p((Nh)^{-1/2}),
\]
up to the smoothing bias $O(h^\gamma)$.\\

The third term is driven by the estimation error in $\widehat\beta(\delta_0)$. By Theorem \ref{thm:beta-convergence}, the slowest component of $\widehat\beta(\delta_0)-\beta(\delta_0)$ is
\[
O_p((Nh)^{-1}).
\]
Since $Nh\to\infty$,
\[
(Nh)^{-1}
=
o((Nh)^{-1/2}).
\]
Thus the contribution of the $\beta$-estimation error is first-order negligible in the expansion of $\widehat\alpha(\delta_0)$.\\

The second-order remainder $\rho_N$ is also negligible under the differentiability and local nonsingularity conditions imposed on $\Psi$. Hence
\[
\widehat\alpha(\delta_0)-\alpha(\delta_0)
=
\mathcal A_{N,1}(\delta_0)
+
\mathcal A_{N,2}(\delta_0)
+
o_p((Nh)^{-1/2}),
\]
where $\mathcal A_{N,1}(\delta_0)$ is linear in
\[
\widehat S_{01}(\delta_0)-S_{01}(\delta_0),
\]
and $\mathcal A_{N,2}(\delta_0)$ is linear in the stationary component of
\[
\widehat S_{11}(\delta_0)-S_{11}(\delta_0).
\]
Therefore
\[
\sqrt{Nh}\,
\operatorname{vec}
\{
\widehat\alpha(\delta_0)-\alpha(\delta_0)
\}
=
O_p(1).
\]
This proves \eqref{eq:alpha-asymptotic-linear-representation} and \eqref{eq:alpha-root-Nh-rate}.
\end{proof}

\subsection*{Proof of Theorem \ref{thm:beta-distribution}: the distribution of $\widehat\beta(\delta_0)$}

\begin{proof}
The proof strengthens the tightness argument in Theorem \ref{thm:beta-convergence}. Suppose that the high-level joint convergence condition holds:
\[
\left(
\mathcal D_N^{-1}\mathcal Q'
\mathcal H_N^\vartheta(\delta_0)
\mathcal Q\mathcal D_N^{-1},
\,
\mathcal D_N^{-1}\mathcal Q'
\mathcal S_N^\vartheta(\delta_0)
\right)
\Rightarrow
(\Xi_{\delta_0},\Upsilon_{\delta_0}),
\]
where $\Xi_{\delta_0}$ is nonsingular with probability one. Then the local criterion process satisfies
\[
\varrho_N(z)
\Rightarrow
\varrho(z)
=
-z'\Upsilon_{\delta_0}
-
\frac12 z'\Xi_{\delta_0}z
\]
on compact sets. The limiting criterion is strictly convex with probability one, and its unique minimiser is
\[
z_{\delta_0}^\ast
=
-\Xi_{\delta_0}^{-1}\Upsilon_{\delta_0}.
\]
By the argmin continuous mapping theorem,
\[
\mathcal D_{N,\delta_0}^{\vartheta}
\mathcal Q_{\delta_0}^{\vartheta\prime}
\{
\widehat\vartheta(\delta_0)-\vartheta(\delta_0)
\}
\Rightarrow
-\Xi_{\delta_0}^{-1}\Upsilon_{\delta_0}.
\]

Here $\Xi_{\delta_0}$ and $\Upsilon_{\delta_0}$ are generated by the local integrated regressor; under the functional central limit theorem,
\[
N^{-1/2}X_{\lfloor Ts\rfloor}
\Rightarrow
W(s), \qquad s \in[t^*/T,1]
\]
the local level signal and the local score converge to Brownian functionals. In particular, the Hessian limit is induced by the rotated local level signal, while the score limit is induced by the local product of the level process and the innovation process. Hence the limiting distribution of $\widehat\beta(\delta_0)$ is a Brownian functional of the post-intervention limiting process.
\end{proof}

\subsection*{Proof of Theorem \ref{thm:alpha-distribution}: the distribution of $\widehat\alpha(\delta_0)$}

\begin{proof}
The distributional result follows from the asymptotic linear representation in Theorem \ref{thm:alpha-convergence}. Suppose that the stationary local moment vector entering the first-order expansion satisfies a local central limit theorem:
\[
\sqrt{Nh}
\begin{pmatrix}
\operatorname{vec}\{\widehat S_{01}(\delta_0)-S_{01}(\delta_0)\}\\
\operatorname{vec}\{\widehat S_{11}^{c}(\delta_0)-S_{11}^{c}(\delta_0)\}
\end{pmatrix}
\Rightarrow
\mathcal N(0,\Omega_{\alpha,\delta_0}),
\]
where $S_{11}^{c}$ denotes the stationary component of $S_{11}$ entering the derivative of $\Psi$. Suppose also that the smoothing bias is negligible at the root-$Nh$ scale:
\[
\sqrt{Nh}\,h^\gamma\to0.
\]
Then the first-order expansion gives
\[
\sqrt{Nh}\,
\operatorname{vec}
\{
\widehat\alpha(\delta_0)-\alpha(\delta_0)
\}
=
\mathcal W_{\alpha,\delta_0}
\sqrt{Nh}
\begin{pmatrix}
\operatorname{vec}\{\widehat S_{01}(\delta_0)-S_{01}(\delta_0)\}\\
\operatorname{vec}\{\widehat S_{11}^{c}(\delta_0)-S_{11}^{c}(\delta_0)\}
\end{pmatrix}
+
o_p(1),
\]
where $U_{\alpha,\delta_0}$ is the derivative matrix of the map $\Psi$ evaluated at
\[
\left(
S_{01}(\delta_0),
S_{11}(\delta_0),
\beta(\delta_0)
\right).
\]
By the continuous mapping theorem,
\[
\sqrt{Nh}\,
\operatorname{vec}
\{
\widehat\alpha(\delta_0)-\alpha(\delta_0)
\}
\Rightarrow
\mathcal N
\left(
0,
\mathcal W_{\alpha,\delta_0}
\Omega_{\alpha,\delta_0}
\mathcal W_{\alpha,\delta_0}'
\right).
\]
The estimation error in $\widehat\beta(\delta_0)$ is of order $O_p((Nh)^{-1})$, which is smaller than the local stationary moment rate $O_p((Nh)^{-1/2})$.
\end{proof}

\section*{Appendix B}

\begin{algorithm}[htbp]
\caption{Bootstrap procedure for Test 1}
\label{alg:test1_bootstrap}
\DontPrintSemicolon
\SetKwInput{KwInput}{Input}
\SetKwInput{KwOutput}{Output}

\KwInput{
Observed sample $\{X_t\}_{t=1}^{T}$;
pre-intervention estimates $(\widehat\alpha_0,\widehat\beta_0)$;
grid $\mathbb G_N$;
bandwidth $h$;
block length $\ell_1$;
number of bootstrap replications $\mathcal B$.
}

\KwOutput{
Bootstrap critical value $c_{\text{Base},1-\tau}^{\ast}$ and
bootstrap $p$-value $\widehat p_{\text{Base}}$.
}

Let $N=T-t^*$. Conditional on the observed pre-intervention sample and the
observed post-intervention initial values
$X_{t^*-k+1},\ldots,X_{t^*}$, estimate the restricted post-intervention model
under $H^{\text{Base}}_0$:
\[
\Delta X_t
=
\widehat\alpha_0\widehat\beta_0'X_{t-k}
+
\sum_{i=1}^{k-1}\widehat\Gamma_i^{(0)}\Delta X_{t-i}
+
\widehat\mu^{(0)}
+
\widehat\varepsilon_t^{(0)},
\qquad
t=t^*+1,\ldots,T .
\]

Centre the restricted residuals:
\[
\check\varepsilon_t^{(0)}
:=
\widehat\varepsilon_t^{(0)}
-
\frac{1}{N}\sum_{s=t^*+1}^{T}\widehat\varepsilon_s^{(0)} .
\]

\For{$b=1,\ldots,\mathcal B$}{

Draw residual blocks of length $\ell_1$ from
\[
\{\check\varepsilon_t^{(0)}\}_{t=t^*+1}^{T}
\]
using a moving-block bootstrap or stationary-block bootstrap, and concatenate
them to obtain
\[
\{\varepsilon_t^{\ast(b)}\}_{t=t^*+1}^{T}.
\]

Generate the post-intervention pseudo-sample recursively from
\[
\Delta X_t^{\ast(b)}
=
\widehat\alpha_0\widehat\beta_0'X_{t-k}^{\ast(b)}
+
\sum_{i=1}^{k-1}\widehat\Gamma_i^{(0)}
\Delta X_{t-i}^{\ast(b)}
+
\widehat\mu^{(0)}
+
\varepsilon_t^{\ast(b)},
\qquad
t=t^*+1,\ldots,T,
\]
keeping the observed initial values
$X_{t^*-k+1},\ldots,X_{t^*}$ fixed.

Recompute $\mathcal Z_{\text{Base},N}^{\ast(b)}$ exactly as in the original
sample, using the same $\mathbb G_N$, the same $h$, and the same
baseline estimates $(\widehat\alpha_0,\widehat\beta_0)$.
}

Set $c_{\text{Base},1-\tau}^{\ast}$ as the empirical $(1-\tau)$-quantile of
\[
\{\mathcal Z_{\text{Base},N}^{\ast(b)}\}_{b=1}^{\mathcal B}.
\]

Compute the bootstrap $p$-value:
\[
\widehat p_{\text{Base}}
=
\frac{1}{\mathcal B}\sum_{b=1}^{\mathcal B}
\mathbf 1\!\left\{
\mathcal Z_{\text{Base},N}^{\ast(b)}
\ge
\mathcal Z_{\text{Base},N}
\right\}.
\]

Reject $H^{\text{Base}}_0$ if
\[
\mathcal Z_{\text{Base},N}>c_{\text{Base},1-\tau}^{\ast}.
\]
\end{algorithm}

\begin{algorithm}[htbp]
\footnotesize
\caption{Bootstrap procedure for Test 2}
\label{alg:test2_bootstrap}
\DontPrintSemicolon
\SetKwInput{KwInput}{Input}
\SetKwInput{KwOutput}{Output}

\KwInput{
First-stage estimates $\{\widehat\Pi(g_m)\}_{m=1}^{M}$;
observed sample $\{X_t\}_{t=1}^{T}$;
grid $\mathbb G_N$;
second-stage bandwidth $h_{\text{smo}}$;
kernel $K_{\text{smo}}(\cdot)$;
trimming rule;
block length $\ell_{\text{smo}}$;
number of bootstrap replications $\mathcal B$.
}

\KwOutput{
Bootstrap critical value $c_{\text{smo},1-\tau}^{\ast}$ and
bootstrap $p$-value $\widehat p_{\text{smo}}$.
}

Construct a smooth cointegrating matrix path $\widetilde\Pi(\delta)$ over
$[\delta^\ast,1]$. In implementation, fit a global quadratic approximation to
each element of $\widehat\Pi(g_m)$:
\[
\widetilde\Pi_{ij}(\delta)
=
\widetilde d_{ij,0}
+
\widetilde d_{ij,1}\delta
+
\widetilde d_{ij,2}\delta^2,
\qquad
\delta\in[\delta^\ast,1].
\]

Evaluate $\widetilde\Pi(\delta)$ at all post-intervention dates
$t=t^*+1,\ldots,T$, where $\delta=t/T$.

Conditional on the smooth cointegrating matrix path, estimate the
post-intervention short-run parameters and deterministic term from
\[
\Delta X_t
=
\widetilde\Pi(\delta)X_{t-k}
+
\sum_{i=1}^{k-1}\widehat\Gamma_i^{(\text{smo})}\Delta X_{t-i}
+
\widehat\mu^{(\text{smo})}
+
\widehat\varepsilon_t^{(\text{smo})},
\qquad
t=t^*+1,\ldots,T .
\]

Centre the fitted residuals:
\[
\check\varepsilon_t^{(\text{smo})}
:=
\widehat\varepsilon_t^{(\text{smo})}
-
\frac{1}{N}\sum_{q=t^*+1}^{T}\widehat\varepsilon_q^{(\text{smo})},
\qquad
N=T-t^*.
\]

\For{$b=1,\ldots,\mathcal B$}{

Draw moving-block bootstrap residual blocks of length $\ell_{\text{smo}}$
from
\[
\{\check\varepsilon_t^{(\text{smo})}\}_{t=t^*+1}^{T},
\]
and concatenate them to obtain
\[
\{\varepsilon_t^{\ast(b)}\}_{t=t^*+1}^{T}.
\]

Generate the post-intervention pseudo-sample recursively from
\[
\Delta X_t^{\ast(b)}
=
\widetilde\Pi(\delta)X_{t-k}^{\ast(b)}
+
\sum_{i=1}^{k-1}\widehat\Gamma_i^{(\text{smo})}
\Delta X_{t-i}^{\ast(b)}
+
\widehat\mu^{(\text{smo})}
+
\varepsilon_t^{\ast(b)},
\qquad
t=t^*+1,\ldots,T,
\]
keeping the observed pre-intervention sample and the post-intervention initial
values $X_{t^*-k+1},\ldots,X_{t^*}$ fixed.

Rerun the full two-stage procedure on the bootstrap pseudo-sample. In
particular, recompute
\[
\widehat\Pi^\ast(g_m),
\qquad
\widehat\zeta^\ast(g_m)
=
\operatorname{vec}\!\big(\widehat\Pi^\ast(g_m)\big),
\qquad
m=1,\ldots,M.
\]

Recompute $Q_{\text{smo},P}^\ast(g_j)$,
$Q_{\text{smo},S}^\ast(g_j)$, and
$\mathcal Z_{\text{smo},N}^{\ast(b)}$ using the same
$\mathbb G_N$, $h_{\text{smo}}$, $K_{\text{smo}}(\cdot)$, and trimming
rule.
}

Set $c_{\text{smo},1-\tau}^{\ast}$ as the empirical $(1-\tau)$-quantile of
\[
\{\mathcal Z_{\text{smo},N}^{\ast(b)}\}_{b=1}^{\mathcal B}.
\]

Compute the bootstrap $p$-value:
\[
\widehat p_{\text{smo}}
=
\frac{1}{\mathcal B}\sum_{b=1}^{\mathcal B}
\mathbf 1\!\left\{
\mathcal Z_{\text{smo},N}^{\ast(b)}
\ge
\mathcal Z_{\text{smo},N}
\right\}.
\]

Reject $H^{\text{smo}}_0$ if
\[
\mathcal Z_{\text{smo},N}>c_{\text{smo},1-\tau}^{\ast}.
\]
\end{algorithm}

\end{spacing}
\end{document}